\documentclass[12pt]{amsart}
\usepackage{network_bootstrap}

\defcitealias{Kojevnikov/Marmer/Song:20}{KMS}

\title{}
\begin{document}

\begin{center}
	\Large \textsc{The Bootstrap for Network Dependent Processes}
	\bigskip
\end{center}

\date{\today}
\email{\href{mailto:D.Kojevnikov@tilburguniversity.edu}{D.Kojevnikov@tilburguniversity.edu}}

\vspace*{5ex minus 1ex}
\begin{center}
	\textsc{Denis Kojevnikov} \\
\bigskip
	\textit{Tilburg University}
	\bigskip
\end{center}

\begin{abstract}
	{\footnotesize This paper focuses on the bootstrap for network dependent processes under the conditional $\psi$-weak dependence. Such processes are distinct from other forms of random fields studied in the statistics and econometrics literature so that the existing bootstrap methods cannot be applied directly. We propose a block-based approach and a modification of the dependent wild bootstrap for constructing confidence sets for the mean of a network dependent process. In addition, we establish the consistency of these methods for the smooth function model and provide the bootstrap alternatives to the network heteroskedasticity-autocorrelation consistent (HAC) variance estimator. We find that the modified dependent wild bootstrap and the corresponding variance estimator are consistent under weaker conditions relative to the block-based method, which makes the former approach preferable for practical implementation.
	}

	\bigskip

	{\footnotesize \noindent \textsc{Keywords.} Conditional bootstrap; Block bootstrap; Dependent wild bootstrap; Network dependent process; Random field; Conditional $\psi$-weak dependence.}

\end{abstract}

\maketitle

\section{Introduction}
The aim of this paper is developing bootstrap approaches for the sample mean of network dependent processes studied in \citet*[hereafter KMS]{Kojevnikov/Marmer/Song:20}. A network dependent process is a random field indexed by the set of nodes of a given undirected network. This network governs the stochastic dependence between the elements of the associated random field. Specifically, the latter is assumed to satisfy a conditional version of the $\psi$-weak dependence condition of \cite{Doukhan/Louhichi:99} given a common shock of a general form. \citealias{Kojevnikov/Marmer/Song:20} show that the pointwise Law of Large Numbers and the Central Limit Theorem hold for a sequence of such processes under suitable assumptions on the networks' denseness and the strength of the stochastic dependence. In addition, they provide nonparametric HAC estimators of the variance-covariance matrix for the vector of sample moments, which is similar to the spatial HAC estimator developed in \cite{Kelejian/Prucha:07}.

These results provide an asymptotic approximation of the distribution of the sample mean which can be used for inference on the true mean of a network dependent process. However, this approximation relies on the network HAC estimator which has two major drawbacks. First, unlike its spatial or time-series counterparts it is not guaranteed to yield a positive semi-definite estimate. Second, these estimators are known to have poor finite sample properties \citep[see, e.g.,][Section~3.5]{Matyas:99:GMM}. The aim of the current work is to provide an alternative nonparametric way to conduct inference in these settings.

The nonparametric bootstrap methods for the case of weakly dependent observations have been studies since the introduction of the non-overlapping block bootstrap in \cite{Carlstein:86} and the moving block bootstrap in \cite{Kunsch:89} and \cite{Liu/Singh:92} for stationary, mixing time-series. Since then, a number of block-based methods have been considered in the statistics literature. They share the idea of resampling groups of consecutive observations to capture the stochastic dependence in the original series and include, among others, the circular block bootstrap \citep[]{Politis/Romano:92}, the stationary bootstrap \citep[]{Politis/Romano:94} and the tapered block bootstrap \citep[]{Paparoditis/Politis:01}. A detailed exposition and comparison of some of these methods can be found in \cite{Lahiri:03:Resampling}. Block-based bootstrap was also successfully applied to the case of weakly dependent random fields satisfying certain mixing conditions \citep[see, e.g.,][Section~12 and references therein]{Lahiri:03:Resampling}.

More recent developments in this area of research are discussed in \cite{Goncalves:11}. In particular, the dependent wild bootstrap proposed in \cite{Buhlmann:93} and \cite{Shao:10} departs from other methods. Instead of using blocks, it tries to mimic the autocovariance structure of the original data by introducing auxiliary random variables and can be applied to irregularly spaced data. A related method, the dependent random weighting, was recently introduced in \cite*{Sengupta/Shao:15} and has wider applicability; specifically, it can be directly applied to irregularly spaced spatial data.

Another useful resampling technique developed for stationary and nonstationary times-series and homogenous random fields under mixing is subsampling. A comprehensive treatment of this method is given in \cite*{Politis:99}. Interestingly, in the time-series case subsampling is similar to the moving block bootstrap where a single block is resampled. Finally, it is worth mentioning the spatial smoothed bootstrap suggested in \cite*{Garcia:13}. In this instance, assuming homogeneity of the underlying data generating process, bootstrap pseudo-samples are drawn from the estimated joint distribution of a given sample.

Network dependent processes are closely related to random fields indexed by elements of a lattice in $\R^d$ \citep[see, e.g.,][]{Comets/Janzura:98,Conley:99}. However, they are not a special case of the latter and so the existing bootstrap methods cannot be directly applied to our framework. The main reason for that is the irregularity of the structure of underlying networks. In particular, subsampling and all types of the block bootstrap for time-series and spatial data rely on the existence of ordered blocks of closely-located observations. The dependent wild bootstrap uses a well-known property of kernel functions that guarantees the positive semi-definiteness of certain weighting matrices. However, as argued in \citealias{Kojevnikov/Marmer/Song:20}, this relation does not necessarily hold when applied to networks. Finally, the homogeneity assumption of the spatial smoothed bootstrap and the spatial subsampling, that is the invariance of joint distributions under spatial shifts is not suitable for our case.

We propose two bootstrap approaches for constructing asymptotically valid confidence sets for the mean of a network dependent process and establish the first-order consistency of these methods for smooth functions of means conditionally on the common shock. The first approach is a block-based method in which blocks are constructed from certain neighborhoods of each node in a network. The second is a modification of the dependent wild bootstrap that employs the topology of a given graph to generate random weights instead of using a fixed kernel function. In addition, we provide the bootstrap variance estimators of the scaled sample mean which yield positive semi-definite estimates and can be used as an alternative to the network HAC estimator. We find that the consistency of the modified dependent wild bootstrap and the corresponding variance estimator holds under weaker conditions as compared with the block bootstrap. However, the bootstrap distribution corresponding to the former method may fail to match the higher-order cumulants of the underlying data generating process, thus preventing improvements over asymptotic approximations.

The rest of the paper is organized as follows. The next section describes a modification of network dependent processes allowing for weighted networks. This modification can be useful for handling dense graphs once varying intensity of links is assumed. Section \ref{sec:cond_bootstrap} provides some general result regarding the conditional bootstrap. Specifically, we use the almost sure convergence of probability kernels to ensure that the bootstrap is valid for (almost) every realization of the common shock, which may also represent the stochastic network formation process. In Section \ref{sec:bootstrap_mean} we present the above-mentioned bootstrap methods in detail and establish sufficient conditions for their conditional consistency.
All the proofs and other technical details are presented in the Appendices \ref{sec:app_main}-\ref{sec:app_aux}.

\section{The Setup}
\label{sec:setup}

We consider a variation of network dependent processes characterized in \citealias{Kojevnikov/Marmer/Song:20}. Namely, let $G\equiv(N,E)$ be an undirected graph (possibly infinite), where $N$ is the set of nodes and $E$ denotes the set of links (we identify $N$ with integers $\{1,2,\ldots\}$). Each edge $e\in E$ is associated with a weight $W(e)\in\bar{\R}$. Also let the function $d:N\times N\to\bar{\R}_{\ge 0}$ be a distance on $G$; for example, the shortest path distance for an unweighted graph. An $\X$-valued \textit{network dependent process} $Y\equiv (Y,G)$ is a collection of $\X$-valued random elements defined on a common probability space indexed by $N$, i.e., $\{Y_i:i\in N\}$. The network $G$ governs the stochastic dependence between random elements. In this paper we consider $\X=\R^v$ with $v\ge 1$.

Further, suppose that we observe a sequence of network dependent processes $\seq{(Y_n,G_n)}$ defined on a common probability space $(\Omega,\H,\PM)$, where each $G_n\equiv(N_n,E_n)$ is a finite graph of size $m_n\to\infty$ as $n\to\infty$; w.l.o.g. we set $m_n=n$. Here, the sequence $\seq{G_n}$ can be a sequence of subgraphs of an infinite network $(N_{\infty},E_{\infty})$. In general, however, these graphs can be unrelated. In order to emphasize the dependence of the distance between two nodes on $n$, we denote it as $d_n(\csdot,\csdot)$. Additionally, since the distance function may implicitly depend on weights associated with the edges of a graph, we impose the following restriction in order to employ the results established in \citealias{Kojevnikov/Marmer/Song:20} with the least possible change.

\begin{assumption}
	\label{assu:min_dist}
	For all $n\ge 1$, $\min_{i,j\in N_n}d_n(i,j)\ge 1$ and $d_n(i,j)=\infty$ whenever $i,j\in N_n$ are disconnected (i.e., there is no path connecting $i$ and $j$).
\end{assumption}

For example, if $W(e)\in [0,1]$ for all $e\in E$, which can be interpreted as the intensity of links, then the shortest weighed distance associated with $1/W(\csdot)$ satisfies this assumption, where implicitly we set $1/0\equiv\infty$. In this case an unweighted network $(N,E)$ is equivalent to a complete graph $(N,E')$, where for $e\in E'$, $W(e)=\ind\{e\in E\}$. In a similar manner, the (at most countable) parameter space of a random field on a metric space $(\mathcal{Z},\rho)$ can be modelled as a compete graph of suitable cardinality, where $W(x\leftrightarrow y)$ is a function of the distance $\rho$ between two points $x,y\in \mathcal{Z}$. Then Assumption \ref{assu:min_dist} corresponds to the case of increasing domain asymptotics \citep[see, e.g.,][]{Conley:99,Jenish/Prucha:09}.

Let $\CS\subset \H$ be a given sub-$\sigma$-field. We assume that the sequence of network dependent processes is conditionally weakly dependent given $\CS$. Specifically, for $a,b\in\N$ and $s\ge 0$ let
\[
	\PS_n(a,b;s)\eqdef \{(A,B)\subset N_n^2: \abs{A}=a,\abs{B}=b,d_n(A,B)\ge s\}
\]
with $d_n(A,B)\eqdef\min_{i\in A,j\in B}d_n(i,j)$ and let $\L_v$ be the family of real-valued, bounded, Lipschitz functions, i.e.,

\[
	\L_v\eqdef \bigcup_{a\ge 1}\L_{v,a},
\]
where
\[
	\L_{v,a}\eqdef \{f:\R^{v\times a}\to \R: \norm{f}_{\infty}<\infty, \Lip(f)<\infty\}.
\]
The functions in $\L_{v,a}$ are Lipschitz with respect to the distance $\delta_a$ on $\R^{v\times a}$ given by
\[
	\delta_a(\vec{x},\vec{y})\eqdef\sum_{l=1}^a\norm{x_l-y_l},
\]
where $\norm{\csdot}$ is a norm on $\R^v$ and $\vec{x}\equiv(x_1,\ldots,x_a)$ and $\vec{y}\equiv(y_1,\ldots,y_a)$ are points in $\R^{v\times a}$. In addition, for a set of nodes $A\subset N_n$ we write $Y_{n,A}\equiv\{Y_{n,i}:i\in A\}$.

\begin{definition}
\label{def:weak_dep}
A sequence $\seq{(Y_n,G_n)}$ is $(\L_v,\psi,\CS)$-weakly dependent if for each $n\ge 1$ there exist a $\CS$-measurable sequence $\gamma_n\equiv\seq{\gamma_{n,s}}_{s=1}^{\infty}$ and a collection of nonrandom functions $(\psi_{a,b})_{a,b\in\N}$, $\psi_{a,b}:\L_{a}\times \L_{b}\to \R_{\ge 0}$ such that for any $(A,B)\in \mathcal{P}_n(a,b;s)$ with $s\ge 1$, $f\in \L_{v,a}$ and $g\in \L_{v,b}$,
\begin{equation}
\label{eq:cov_bound}
	\abs{\Cov(f(Y_{n,A}),g(Y_{n,B})\mid \CS)}\le \psi_{a,b}(f,g)\gamma_{n,\floor{s}} \qtext{a.s.}
\end{equation}
\end{definition}

\begin{remark*}
(a) When it is clear from the context, we denote such a sequence as $\seq{Y_n}$ omitting the reference to the underlying networks. (b) $(\L_v,\psi)\equiv (\L_v,\psi,\{\emptyset,\Omega\})$. (c) The elements of $\arr{\gamma_{n,s}}$ are called the weak-dependence coefficients associated with $\seq{Y_n}$. (d) For convenience, we set $\gamma_{n,0}\equiv 1$.
\end{remark*}

A number of examples of network dependent processes that are $(\L_v,\psi,\CS)$-weakly dependent are given in \citealias{Kojevnikov/Marmer/Song:20}. For instance, strong mixing processes correspond to $\psi_{a,b}(f,g)=4\norm{f}_{\infty}\norm{g}_{\infty}$. Also associated and Gaussian processes and their certain derivatives are $(\L_v,\psi,\CS)$-weakly dependent with $\psi_{a,b}(f,g)=ab\Lip(f)\Lip(g)$. It is worth mentioning that the corresponding weak dependence coefficients may depend on the topology of the underlying networks.

Conditioning on a $\sigma$-field $\CS$ can be useful in various cases. First, if the underlying graphs are realizations of a stochastic network formation process, then one can potentially condition on the $\sigma$-field generated by that process and treat the observed graphs as fixed. Second, fixing nodes with high degree centrality may help to obtain local stochastic dependence.

\begin{example}
Consider a set independent random variables $\seq{\varepsilon_i:i\in N}$ and let $C\subset N$ denote a set of nodes with ``high'' degree centrality (for clarity, we omitted the subscript $n$). Then $u_{N\setminus C}$, where $u_i\eqdef \varepsilon_i +\sum_{j\in C}\beta_{ij}\varepsilon_j$ and $\beta_{ij}\in \R$, are conditionally independent given $\CS=\sigma(\varepsilon_C)$. Moreover, for arbitrary measurable functions $\seq{\phi_i}$ the process $\seq{Y_i\eqdef\phi_i(u_N)}$ satisfies the covariance bound \eqref{eq:cov_bound} with $\psi_{a,b}(f,g)=a\norm{g}_{\infty}\Lip(f)+b\norm{f}_{\infty}\Lip(g)$. In the context of social interaction models $\seq{u_i}$ and $\seq{Y_i}$ may represent idiosyncratic shocks and observable outcomes, respectively.
\end{example}

In order to facilitate the exposition, throughout the paper we consider a sequence of network dependent processes $\seq{Y_n}$ satisfying the covariance bound \eqref{eq:cov_bound} with a specific form of the function $\psi_{a,b}$ and bounded weak dependence coefficients. The restricted $\psi_{a,b}$ function is fairly general and covers many useful examples of weakly dependent processes.

\begin{assumption}
\label{assu:psi_weak_dep}
$\seq{(Y_n,G_n)}$ is $(\L_v,\psi,\CS)$-weakly dependent and there exist constants $M\ge 1$ and $C>0$ such that $\gamma_{n,s}\leq M$ a.s.\ for all $n,s\ge 1$, and
\begin{align*}
	\psi_{a,b}(f,g)&=c_1\norm{f}_{\infty}\norm{g}_{\infty}\,\,+c_2\Lip(f)\norm{g}_{\infty} \\
	&\quad +c_3\norm{f}_{\infty}\Lip(g)+c_4\Lip(f)\Lip(g),
\end{align*}
where $c_1,\ldots,c_4\le Cab$.
\end{assumption}

It should be noted that processes satisfying Assumption \ref{assu:psi_weak_dep} possess some hereditary properties. Specifically, if $\seq{Y_n}$ is $(\L_v,\psi,\CS)$-weakly dependent with the weak dependence coefficients $\arr{\gamma_{n,s}}$, then for any Lipschitz function $h:\R^v\to\R^w$ the sequence $\seq{h(Y_{n,i}):i\in N_n}$ is $(\L_w,\psi,\CS)$-weakly dependent with the same weak dependence coefficients. Moreover, this type of weak dependence is preserved under some locally Lipschitz functions as shown in Proposition \ref{prop:local_lipshitz} below, which is an extension of Proposition 2.1. in \citet*{Dedecker/Doukhan:07:WeakDep} to our settings.

\begin{prop}
\label{prop:local_lipshitz}
Suppose that $\seq{Y_n}$ satisfies Assumption \ref{assu:psi_weak_dep} and there exist $L<\infty$ and $p>1$ such that $\sup_{n,i\in N_n}\E[\norm{Y_{n,i}}_{\infty}^p\mid \CS]\le L$ a.s. Let $h:\R^v\to\R^w$ be such that
\begin{equation}
\label{eq:local_lipshitz0}
	\norm{h(x)-h(y)}\le \eta\norm{x-y}\left(\norm{x}^{\tau-1}+\norm{y}^{\tau-1}\right)
\end{equation}
for some $\eta>0$ and $\tau\in [1,p)$.
Then $\seq{h(Y_{n,i}):i\in N_n}$ is $(\L_w,\psi,\CS)$-weakly dependent with the weak dependence coefficients
\[
	\gamma_{n,s}'=KM\gamma_{n,s}^r,
\]
where $K$ is a constants depending on $\eta$, $v$, and $L$ and
\[
	r= \begin{cases}
		(p-\tau)/(p-1), & \text{if } c_4=0, \\
		(p-\tau)/(p+\tau-2), & \text{otherwise}.
	\end{cases}
\]
\end{prop}

\begin{remark*}
The boundedness of the conditional moments of $\norm{Y_{n,i}}_{\infty}$ is required in order to maintain Assumption \ref{assu:psi_weak_dep}. Once this condition is relaxed, it suffices to assume that these moments are a.s.\ finite.
\end{remark*}

\subsection{Asymptotic Results}
Introducing weighted networks is useful in several scenarios. First, as we have already mentioned it allows incorporating some additional random processes into the current framework. Second, assuming varying intensity of connections enables one to handle denser networks in the sense of the total number of links. Finally, some commonly used statistical models explicitly use weights and can be adapted to our framework, e.g., the spatial Cliff-Ord-type linear model in \cite{Kelejan/Prucha:07_2}.

\begin{example}
For each $n\ge 1$, let $u_n$ be a $n\times 1$ vector of independent random variables and let $\tilde{W}_n$ be an $n\times n$ matrix which is a function of weights associated with a given network. Consider a linear model with disturbances following the next autoregressive process:
\[
	\varepsilon_n=\lambda\tilde{W}_n\varepsilon_n+u_n, \quad\abs{\lambda}<1,
\]
Typically the original weighting matrix is modified to ensure that the spectral radius of $\tilde{W}_n$ is bounded by $1$. Under certain restrictions on the denseness of underlying networks, the process $\seq{\varepsilon_n}$ is weakly dependent with $\psi_{a,b}(f, g)=a\norm{g}_{\infty}\Lip(f)+b\norm{f}_{\infty}\Lip(g)$ so that the model can be accommodated within the current framework.

Assume that $C_n\eqdef (I-\lambda\tilde{W}_n)^{-1}$ exists for each $n\ge 1$ and $\mu\eqdef \sup_{n,i\in N_n}\E\abs{u_{n,i}}<\infty$. Then $\varepsilon_n=C_n u_n$ and, letting $\varepsilon_{n,i}^{(s)}\eqdef \sum_{j\in N_n:d_n(i,j)<s+1}[C_n]_{ij}u_{n,j}$,
\begin{align*}
	&\E[\absin{\varepsilon_{n,i}-\varepsilon_{n,i}^{(s)}}]\le \mu \max_{i\in N_n}\sum_{j\in N_n:d_n(i,j)\ge s+1}\abs{[C_n]_{ij}}\equiv \gamma_{n,s}.
\end{align*}
Thus, $\arr{\varepsilon_n}$ is $(\L_1,\psi)$-weakly dependent provided that $\sup_n\gamma_{n,s}\to 0$ as $s\to \infty$. In order to appreciate the magnitude of the elements of $C_n$ consider a simple case when $\tilde{W}_n=A_n/\rho(A_n)$ and $A_n$ is the adjacency matrix of the underlying graph. Then $C_n=\sum_{k\ge 0}(\lambda/\rho(A_n))^kA_n^k$, where $[A_n^k]_{ij}$ measures the number of paths of length $k$ between nodes $i$ and $j$ which is directly related to the denseness of the network.
\end{example}

For a given network $G_n$ let $N_n(i;s)$ denote the open neighborhood of radius $s>0$ around $i\in N_n$, i.e,
\[
	N_n(i;s)\eqdef\{j\in N_n:d_n(i,j)<s\},\footnote{
		Note that this definition of the open neighborhood of a node differs from one commonly used in graph theory.
	}
\]
and let $N_n^{\partial}(i;s)\eqdef N_n(i;s+1)\setminus N_n(i;s)$. In addition, we define the following aggregate measures of the network denseness:
\begin{equation}
\label{eq:network_measures}
	\begin{alignedat}{2}
		\delta_n(s;k)&\eqdef n^{-1}\sum_{i\in N_n}\abs{N_n(i;s+1)}^k, \quad \delta_n^{\partial}(s;k)&&\eqdef n^{-1}\sum_{i\in N_n}\absin{N_n^{\partial}(i;s)}^k, \\
		D_n(s)&\eqdef \max_{i\in N_n}\abs{N_n(i;s+1)}, \qtext{and}\quad D_n^{\partial}(s)&&\eqdef \max_{i\in N_n}\absin{N_n^{\partial}(i;s)}.
	\end{alignedat}
\end{equation}

It is straightforward to see that under Assumption \ref{assu:min_dist}, which restricts the minimum distance between any two nodes of a network, the asymptotic results derived in \citealias{Kojevnikov/Marmer/Song:20} remain valid once we replace their measures of network denseness with those given in \eqref{eq:network_measures} and redefine $H_n(s,m)$ as follows:
\begin{equation}
\label{eq:network_measures2}
	H_n(s,m)\eqdef \left\{(i,j,k,l)\in N_n^4 :
	\scaleobj{0.8}{
		\begin{gathered}
			j\in N_n(i;,m+1),l\in N_n(k;m+1), \\[-1ex]
			\floor{d_n(\{i,j\},\{k,l\})}=s
		\end{gathered}}
	\right\}.
\end{equation}

In the case of random networks, however, the measures of network denseness are also random. Therefore, one needs a conditional version of the Law of Large Numbers in order to be able to condition on the common shock $\CS$. Note that the other result are stated in the conditional form and can be directly applied to this case if we assume certain measurability conditions. Let $\D(G_n)$ denote the distance matrix associated with $G_n$, i.e., $[\D(G_n)]_{ij}=d_n(i,j)$. If $\D(G_n)$ is $\CS$-measurable, then $N_n(i;s)=\sum_{j\in N_n}\ind\{[\D(G_n)]_{ij}<s\}$ is also $\CS$-measurable as well as the quantities given in \eqref{eq:network_measures} and \eqref{eq:network_measures2}. We make the following assumption:

\begin{assumption}
\label{assu:msble_dist}
The distance matrix $\D(G_n)$ is $\CS$-measurable for all $n\ge 1$.
\end{assumption}

In addition, we introduce the notion of asymptotically negligible random functions, which is useful for defining the conditional versions of the asymptotic tightness and uniform integrability.

\begin{definition}
Let $\F\subset \H$ and let $f:\Y\times\Omega\to\R_{\ge 0}$ be such that $f(y,\csdot)$ is $\F$-measurable for all $y\in\Y$. A sequence of such functions $\seq{f_n}$ is \textit{asymptotically negligible} (a.n.), if for almost all $\omega\in \Omega$,
\[
	\essinf_{s\in S}\limsup_{n\to\infty}f_n(s,\omega)=0.\footnote{
		The essential infimum of an arbitrary family of random variables $\{Z_{\alpha}:\alpha\in \mathcal{A}\}$ a random variable $Z$ such that (a) $Z\le Z_{\alpha}$ a.s.\ for all $\alpha\in \mathcal{A}$ and (b) $Z\ge Z'$ a.s.\ for any random variable $Z'$ satisfying $Z'\le Z_{\alpha}$ a.s.\ for all $\alpha\in \mathcal{A}$. In particular, there exists a sequence $\seq{\alpha_n}$ such that $Z=\inf_{n}Z_{\alpha_n}$ a.s.\ \citep[see, e.g.,][Theorem~1.3.40]{Cohen:15:StochCalc}. The essential supremum is defined similarly and the following identity holds: $\esssup_{\alpha}\{Z_{\alpha}\}=-\essinf_{\alpha}\{-Z_{\alpha}\}$.
	}
\]
In particular, an array of random vectors $\arr{X_{n,i}}$ is
\begin{itemize}[leftmargin=*]
	\item[--] $\F$-\textit{asymptotically tight} if $\max_{i}\PR{\norm{X_{n,i}}>y\mid \F}$ is a.n.
	\item[--] $\F$-\textit{asymptotically uniformly integrable} (u.i.) if $\max_{i}\E[\norm{X_{n,i}}\ind\{\norm{X_{n,i}}>y\}\mid \F]$ is a.n.
\end{itemize}
\end{definition}

\begin{theorem}[Conditional Weak Law of Large Numbers]\label{thm:CWLLN}
Let $\seq{(Y_n,G_n)}$ be $(\L_v,\psi,\CS)$-weakly dependent satisfying Assumption \ref{assu:min_dist}, \ref{assu:psi_weak_dep}, and \ref{assu:msble_dist}. Suppose that $\seq{Y_n}$ is $\CS$-asympto-tically u.i. and
\[
	\frac{1}{n}\sum_{s\ge 1}\delta_n^{\partial}(s;1)\gamma_{n,s}\to 0 \qtext{a.s.}
\]
Then
\[
	\norm{\frac{1}{n}\sum_{i\in N_n}(Y_{n,i}-\E[Y_{n,i}\mid \CS])}_{\CS,1}\to 0 \qtext{a.s.}\footnote{
		For a random vector $X$ and $\F\subset\H$ we write $\norm{X}_{\F,p}\equiv \E[\norm{X}^p\mid \F]^{1/p}$.
	}
\]
\end{theorem}

\begin{remark*}
Similarly to the unconditional case a sufficient condition for the $\CS$-asymptotic uniform integrability of $\seq{Y_n}$ is the a.s.\ finiteness of $\sup_{n,i\in N_n}\E[\normin{Y_{n,i}}^p\mid \CS]$ for some $p>1$.
\end{remark*}

Finally, let $\bar{Y}_n\eqdef n^{-1}\sum_{i\in N_n}Y_{n,i}$ and $\Sigma_n\eqdef \Var(\sqrt{n}\bar{Y}_n\mid \CS)$. Then the network HAC estimator of $\Sigma_n$,
\begin{equation}
\label{eq:HAC}
	\hat{\Sigma}_n=\frac{1}{n}\sum_{i,j\in N_n}\kappa\left(\frac{d_n(i,j)}{b_n+1}\right)(Y_{n,i}-\bar{Y}_n)(Y_{n,j}-\bar{Y}_n)^{\top},
\end{equation}
where $\kappa:\bar{\R}\to[-1,1]$ is a kernel function satisfying: $\kappa(0)=1$, $\kappa(z)=\kappa(-z)$, and $\kappa(z)=0$ for $\abs{z}>1$ and $b_n$ is the lag truncation parameter, is consistent under the same set of assumptions. Unfortunately, due to the irregularity of a network's structure, this estimator is not guaranteed to be positive semi-definite. However, once the minimal eigenvalue of $\Sigma_n$ is a.s.\ bounded from below or it converges to an a.s.\ positive definite matrix, a simple way to fix this issue is available. The details are given in Appendix \ref{section:app_HAC}.

\section{Conditional Bootstrap}
\label{sec:cond_bootstrap}

In this section we present some general result regarding the conditional bootstrap. The latter is useful for an inference which is asymptotically valid for almost all $\omega\in \Omega$ (or almost all realizations of the common shock). These results do not depend on the underlying data generating process. However, we use the present framework for convenience.

Suppose that $\seq{(Y_n, G_n)}$ is a sequence of network dependent processes. For a given $n\ge 1$ let $\theta_n$ be a $\CS$-measurable parameter taking values in $\Theta\subseteq \R^w$ with $w\ge 1$ and let
\[
	T_n(\theta_n)\eqdef\T_n(Y_n,\theta_n;\vartheta_n),
\]
where $\T_n$ is a measurable, real-valued function and $\vartheta_n$ is a $\CS$-measurable nuisance parameter, denote a statistic used to conduct inference on $\theta_n$ based on a realization of $(Y_n,G_n)$ conditionally on $\CS$.

Let $\FD{n}{}$ denote the conditional cdf of $T_n$ given $\CS$.\footnote{
	A (regular) conditional cdf $\FD{X}{\F}$ of $X\in \R$ given $\F\subset\H$ satisfies: (i) $\forall x\in \R$, $\FD{X}{\F}(\csdot,x)$ is a version of $\PR{X\le x\mid \F}$, and (ii) $\forall \omega\in \Omega$, $\FD{X}{\F}(\omega,\csdot)$ is a distribution function. We omit the subscript $X$ or the superscript $\F$ whenever clear from the context.
}
The goal of this section is to provide sufficient conditions for the conditional first-order consistency of resampling estimators of $\FD{n}{}$. Specifically, let $\G_n\eqdef \CS \vee \sigma(Y_n)$ and let $Y_n^{*}$ be a pseudo-sample drawn using a realization of $Y_n$. Then the bootstrap counterpart of $T_n$ is $T_n^{*}\eqdef \T_m(Y_n^{*},\theta_n^{*})$, where $m$ is the size of $Y_n^{*}$ and $\theta_n^{*}\equiv \theta_n^{*}(Y_n)$ is an estimator of $\theta_n$. The conditional cdf $\FD{n}{*}$ of $T_n^{*}$ is used as an approximation of $\FD{n}{}$. If the latter explicitly depends on the nuisance parameter $\vartheta_n$, then one needs to provide its consistent estimator based on both $Y_n$ and $Y_n^{*}$.

A typical way of showing the consistency of the bootstrap estimators is bounding the Kolmogorov distance between the cdfs of $T_n$ and $T_n^{*}$ \citep[see, e.g.,][Chapter~3]{Shao:95:Bootstrap}. For random variables $X$ and $Y$ and sub-$\sigma$-fields $\F\subset\G\subset\H$ the conditional version of the latter is defined by
\[
	\DK{X,Y\mid \G,\F}\eqdef \sup_{x\in \R}\abs{\FD{X}{\G}(\csdot,x)-\FD{Y}{\F}(\csdot,x)},\footnote{
		Note that $\DK{\csdot,\csdot\mid \G,\F}$ is $\G$-measurable because $\seq{Z_x}$, where $Z_x\eqdef\abs{\FD{X}{\G}(\csdot,x)-\FD{Y}{\F}(\csdot,x)}$, is a c\'adl\'ag stochastic process).
	}
\]
where $\FD{X}{\G}$ and $\FD{Y}{\F}$ are the conditional cdfs of $X$ and $Y$, respectively (when $\F=\G$ we denote this measure by $\DK{X,Y\mid \F}$). In addition, we define the conditional convergence in probability and the almost sure convergence of conditional distributions.\footnote{
	A (regular) conditional distribution $\QM_X^{\F}$ of $X\in \R^v$ given $\F\subset\H$ satisfies: (i) $\forall B\in \B(\R^v)$, $\QM_X^{\F}(\csdot, B)$ is a version of $\PR{X\in B\mid \F}$ and (ii) $\forall \omega\in\Omega$, $\QM_X^{\F}(\omega,\csdot)$ is a probability measure on $(\R,\B(\R^v))$. We omit the subscript $X$ or the superscript $\F$ whenever clear from the context.
}

\begin{definition}
\label{def:CCP}
Let $\F\subset\H$ be a sub-$\sigma$-field and let $Z$ be a $\F$-measurable random vector in $\R^v$ with $v\ge 1$. A sequence of $\R^v$-valued random vectors $Z_n\PRC{\F}Z \qtext{a.s.}$ if for any $\epsilon>0$, $\PR{\norm{Z_n-Z}>\epsilon\mid \F}\to 0$ a.s.\footnote{
	Note that this implies convergence in probability due to the dominated convergence theorem. In addition, for an a.s.\ positive $\F$-measurable random variable $\nu$, $\PR{\norm{Z_n-Z}>\nu\mid \F}\to 0$ a.s.
}
\end{definition}

\begin{definition}
\label{def:CWC}
Suppose that $\seq{X_n}$ is a sequence of random vectors on $(\Omega,\H,\PM)$ and $\F\subset\H$. Let $\QM_n$ be the conditional distribution of $X_n$ given $\F$.
We say that $X_n$ converges $\F$-weakly to $X$ having the conditional distribution $\QM$ if for almost all $\omega\in\Omega$ the sequence $\seq{\QM_n(\omega,\csdot)}$ converges weakly to $\QM(\omega,\csdot)$.
\end{definition}

\begin{remark*}
(a) Equivalently, the $\F$-weak convergence can be defined using the notion of probability kernels. So the limiting random vector $X$ is an artificial construct which is used to describe the limiting kernel. (b) A more general notion of the almost sure convergence of conditional probability measures and some of its properties are presented in \cite*{Berti:06}. (c) The notion of the $\F$-weak convergence is stronger than the $\F$-stable convergence and the usual weak convergence. In particular, if $X_n\to X$\, $\F$-weakly, then for any real-valued, bounded, continuous function $f$, $\E[f(X_n)\mid \F]\to\E[f(X)\mid \F]$ a.s., which implies that $X_n$ converges to $X$\, $\F$-stably and in distribution.
\end{remark*}

Assume for a moment that $\CS=\{\emptyset,\Omega\}$. Then if there exists a sequence of random variables $\seq{S_n}$ such that $\DK{T_n,S_n\mid \CS}$ converges to 0 as $n\to\infty$ and $\DK{T_n^{*},S_n\mid \G_n,\CS}$ converges to $0$ a.s.\ (in probability), then the bootstrap estimator is first-order strongly (weakly) consistent. Moreover, if $S_n$ converges weakly to a continuous limit, then the conditional quantiles of $\FD{n}{*}$ are a good approximation to those of $\FD{n}{}$. This typically happens when the statistic $T_n$ is pivotal. However, in the case of a non-pivotal statistic, which is useful when a consistent estimator of $\vartheta_n$ is hard to obtain or the available estimators have poor finite sample properties, the cdfs of $\seq{T_n}$ need not converge.\footnote{
	Consider, for example, the case of the linearized statistic $T_n'$ given in \eqref{eq:linear_stat}. It does not have a nondegenerate weak limit when the sequence of parameters $\seq{\theta_n}$ is not convergent.
}
In this case, the convergence of the Kolmogorov distance between $T_n^{*}$ and $T_n$ to zero does not necessarily imply that $\FD{n}{}(c_n^{*}(\alpha))\to \alpha$ as $n\to\infty$, where $c_n^{*}(\alpha)$ is the conditional $\alpha$-quantile of $\FD{n}{*}$. Nevertheless, as shown in the next result, a sufficient condition for the latter to happen is the continuity of the cdfs of $\seq{S_n}$.

\begin{theorem}
\label{thm:bootstrap_consistency1}
Suppose that for all $n\ge 1$, $S_n$ is conditionally independent of $Y_n$ given $\CS$ and the conditional cdf of $S_n$ given $\CS$ is \tn(a.s.\tn) continuous. Then if
\begin{itemize}[leftmargin=1.65em]
	\item[\tn{(a)}] $\DK{T_n,S_n\mid \CS}\to 0$ a.s.\ and
	\item[\tn{(b)}] $\DK{T_n^{*},S_n\mid \G_n,\CS}\PRC{\CS} 0$ a.s.,
\end{itemize}
$\DK{T_n^{*},T_n\mid \G_n,\CS}\PRC{\CS} 0$ a.s.\ and
\[
	\esssup_{\alpha\in(0,1)}\abs{\PR{T_n\le c_n^{*}(\alpha)\mid \CS}-\alpha}\to 0 \qtext{a.s.}
\]
\end{theorem}

\begin{remark*}
(a) Usually when $\CS=\{\emptyset,\Omega\}$ and the statistic $T_n$ is pivotal, we have $S_n=S_{\infty}$, which is the weak limit of $T_n$. (b) A variant of this result can be found in \citet[]{Chernozhukov/Chetverikov/Kato:13} in the context of Gaussian multiplier bootstrap. (c) Theorem \ref{thm:bootstrap_consistency1} also implies that the conditional quantiles $\{c_n^{*}(\alpha):\alpha\in (0,1)\}$ approximate the unconditional quantiles of $T_n$ because, by the dominated convergence theorem,
\[
	\sup_{\alpha\in (0,1)}\abs{\PR{T_n\le c_n^{*}(\alpha)}-\alpha}\le \E\esssup_{\alpha\in(0,1)}\abs{\PR{T_n\le c_n^{*}(\alpha)\mid \CS}-\alpha}\to 0.
\]
\end{remark*}

\begin{definition}
\label{def:bootstrap_consistency}
We say that $\FD{n}{*}$ is conditionally $\DKM$-consistent given $\CS$ if the conclusion of Theorem \ref{thm:bootstrap_consistency1} holds.
\end{definition}

Typically it is not hard to show that condition (a) of Theorem \ref{thm:bootstrap_consistency1} holds (for example, when the elements of $Y_n^{*}$ are conditionally i.i.d.\ given $\G_n$). On contrary, establishing (b) may be a difficult task, especially when $T_n$ is a nonlinear transformation of $Y_n$ in the presence of stochastic dependence between its elements as in the current framework. However, in the case when the statistic $T_n$ converges $\CS$-weakly to $S$ and the limiting kernel (i.e., the regular conditional cdf of $S$ given $\CS$) is continuous, Lemma \ref{lemma:aux_weak_eq} implies that this convergence is equivalent to one with respect to the conditional Kolmogorov distance. In addition, by Lemma \ref{lemma:aux_weak_conv} the almost sure convergence of conditional distributions enjoys a number of useful properties associated with the usual weak convergence such as the continuous mapping theorem, converging together lemma, and the Cram\'er–Wold device. In this situation we have the following simple corollary.

\begin{corollary}
\label{corr:bootstrap_consistency1}
Suppose that $S$ is conditionally independent of $\seq{Y_n}$ given $\CS$ and the conditional cdf of $S$ given $\CS$ is \tn(a.s.\tn) continuous. Then if
\begin{itemize}[leftmargin=1.65em]
	\item[\tn{(a)}] $T_n\to S$ \, $\CS$-weakly and
	\item[\tn{(b)}] $\DK{T_n^{*},S\mid \G_n,\CS}\PRC{\CS} 0$ a.s.,
\end{itemize}
$\FD{n}{*}$ is conditionally $\DKM$-consistent given $\CS$.
\end{corollary}

Next, we consider the case in which the statistic $T_n$ takes the following form:
\[
	T_n(\theta_n)=\tau_n\left(\phi(\hat{\theta}_n)-\phi(\theta_n)\right),
\]
where $\phi:\Theta\to \R$ is a continuously differentiable function, $\hat{\theta}_n$ is a consistent estimator of $\theta_n$ (in the sense of Definition \ref{def:CCP}) and $\tau_n$ is a normalizing coefficient. In particular, the \textit{smooth function model} (see, e.g., \citealp[Section~4.2]{Lahiri:03:Resampling} and \citealp[Section~2.4]{Hall:92:Bootstrap}) falls into this case. The resampling version of the statistic $T_n$ is
\[
	T_n^{*}=\tau_n^{*}\left(\phi(\hat{\theta}_n^{*})-\phi(\theta_n^{*})\right),
\]
where $\theta_n^{*}$ is a consistent estimator of $\theta_n$, which may differ from $\hat{\theta}_n$, and $\tau_n^{*}$ is the bootstrap counterpart of $\tau_n$. Let $\xi_n\eqdef \tau_n(\hat{\theta}_n-\theta_n)$ and $\xi_n^{*}\eqdef \tau_n^{*}(\hat{\theta}_n^{*}-\theta_n^{*})$. Consider the linearized statistics
\begin{equation}
\label{eq:linear_stat}
	T_n'\eqdef \nabla\phi(\theta_n)^{\top}\zeta_n \qtext{and}\quad T_n'^{*}\eqdef \nabla\phi(\theta_n^{*})^{\top}\xi_n^{*}.
\end{equation}
The following result shows that it suffices to find a ``smooth'' approximation $S_n'$ of the linearized statistics in order to apply Theorem \ref{thm:bootstrap_consistency1} to this setup. In particular, the result largely depends on the asymptotic behavior of the conditional L\'evy concentration function of $S_n'$. For a random variable $X$, $\epsilon>0$, and a sub-$\sigma$-field $\F\subset\H$ the latter is given by
\[
	\LC{\epsilon,X\mid \F}\eqdef \sup_{x\in \R}\left(\FD{X}{\F}(\csdot,x+\epsilon)-\FD{X}{\F}(\csdot,x-)\right).
\]

\begin{lemma}
\label{lemma:bootstrap_consistency2}
Suppose that $\hat{\theta}_n^{*}-\theta_n^{*}\PRC{\CS}0$ a.s., $\xi_n^{*}$ and $\xi_n$ are $\CS$-asymptotically tight and $\sup_n\norm{\theta_n}<\infty$ a.s. Furthermore, assume that
\begin{itemize}[leftmargin=1.65em]
	\item[\tn{(a)}] $\DK{T_n',S_n'\mid \CS}\to 0$ a.s.,
	\item[\tn{(b)}] $\DK{T_n'^{*},S_n'\mid \G_n,\CS}\PRC{\CS}0$ a.s.\ and
	\item[\tn{(c)}] $\LC{\epsilon,S_n'\mid \CS}$ is a.n.
\end{itemize}
Then w.p.1,
\[
	\DK{T_n^{*},S_n'\mid \G_n,\CS}\PRC{\CS}0 \qtext{and}\quad \DK{T_n,S_n'\mid \CS}\to 0.
\]
\end{lemma}

Consequently, the continuity of the conditional cdfs of $\seq{S_n'}$ ensures the bootstrap consistency in the sense of Definition \ref{def:bootstrap_consistency}.

\begin{theorem}
\label{thm:bootstrap_consistency2}
Suppose that the conditions of Lemma \ref{lemma:bootstrap_consistency2} hold and, in addition, $\seq{S_n}$ satisfy the independence and continuity conditions of Theorem \ref{thm:bootstrap_consistency1}. Then $\FD{n}{*}$ is conditionally $\DKM$-consistent given $\CS$.
\end{theorem}

Similarly to the general case, when $\xi_n$ converges $\CS$-weakly to some random vector $\xi$ and the sequence of parameters $\seq{\theta_n}$ converges a.s.\ to a $\CS$-measurable random variable $\theta$, Lemma \ref{lemma:aux_weak_conv} implies that $T_n'$ converges $\CS$-weakly to $\nabla\phi(\theta)^{\top}\xi$. In addition, if the conditional cdf of the latter is (a.s.) continuous, it satisfies assumption (b) of Lemma \ref{lemma:bootstrap_consistency2}.

\begin{corollary}
\label{corr:bootstrap_consistency2}
Suppose that $\hat{\theta}_n^{*}-\theta_n^{*}\PRC{\CS}0$ a.s., $\xi_n^{*}$ is $\CS$-asymptotically tight and $S'\eqdef\nabla\phi(\theta)^{\top}\xi$ satisfies the independence and continuity conditions of Corollary \ref{corr:bootstrap_consistency1}. Then if
\begin{itemize}[leftmargin=1.65em]
	\item[\tn{(a)}] $\xi_n\to \xi$ \, $\CS$-weakly,
	\item[\tn{(b)}] $\theta_n\to\theta$ a.s.\ and
	\item[\tn{(c)}] $\DK{T_n'^{*},S'\mid \G_n,\CS}\PRC{\CS}0$ a.s.,
\end{itemize}
$\FD{n}{*}$ is conditionally $\DKM$-consistent given $\CS$.
\end{corollary}

\begin{remark*}
The assumption regarding convergence of the sequence of parameters $\seq{\theta_n}$ can be relaxed. In the unconditional case it suffices to assume that $\sup_n\norm{\theta_n}<\infty$. Then one needs to provide a uniform bound on $\abs{\PR{\xi_n\in A}-\PR{\xi\in A}}$, where $A$ ranges over the class of half-spaces for a network dependent process similar to that established in \cite{Bentkus:03}. The conditioning on $\CS$ complicates the problem even more so it falls out of the scope of this paper.
\end{remark*}

\section{Bootstrap of the Mean}
\label{sec:bootstrap_mean}

Consider a sequence of network dependent processes $\seq{(Y_n, G_n)}$ satisfying Assumptions \ref{assu:min_dist}, \ref{assu:psi_weak_dep}, and \ref{assu:msble_dist}. As an application of the results given in the preceding section, we consider the mean of a $Y_n$, $\mu_n\equiv \E[Y_{n,i}\mid \CS]$ which may vary with $n$ but not across $i\in N_n$.\footnote{
	It is possible to extend the results of this paper to the case of heterogeneous means as in \cite{Goncalves/White:02}. However, such a setup makes it difficult to isolate the effect of the structure of underlying networks on the consistency of the proposed bootstrap methods.
}
The parameter of interest $\mu_n$ is estimated using the sample mean $\bar{Y}_n$ which is a consistent estimator of $\mu_n$ under the assumptions of Theorem \ref{thm:CWLLN}. In this section we provide a number of resampling based methods for constructing the asymptotically valid confidence sets for $\mu_n$. In addition, we establish consistency of a restricted version of the \textit{smooth function model} in which we are interested in $\phi(\mu_n)$ for a continuously differential function $\phi:\R^v\to \R$. When the elements of $Y_n$ have the same marginal conditional distributions given $\CS$, we may consider $\phi(\E[f(Y_{n,1})\mid \CS])$, where $f:\R^v\to \R^w$ is a locally Lipschitz function satisfying \eqref{eq:local_lipshitz0} and the domain of $\phi$ is $\R^w$ in this case. Since the process $\seq{f(Y_{n,i}):i\in N_n}$ is $(\L_w,\psi,\CS)$-weakly dependent by Proposition \ref{prop:local_lipshitz}, without loss of generality we examine the first version. In addition, we provide consistent positive semi-definite estimators of $\Sigma_n$.

The corresponding test statistics are given by
\begin{align}
\label{eq:test_stats}
\begin{aligned}
	T_{1,n}(\mu_n)&=\sqrt{n}\normin{\bar{Y}_n-\mu_n}, \qtext{and} \\
	T_{2,n}(\mu_n)&=\sqrt{n}\big(\phi(\bar{Y}_n)-\phi(\mu_n)\big),
\end{aligned}
\end{align}
where $\norm{\csdot}$ is the Euclidean norm on $\R^v$. Their conditional distributions given $\CS$ are denoted by $\FD{1,n}{}$ and $\FD{2,n}{}$, respectively, and the bootstrap approximations of these distributions are denoted by $\FD{1,n}{*}$ and $\FD{2,n}{*}$. The confidence sets for $\mu_n$ are obtained by test inversion, i.e
\[
	CS_{n,1-\alpha}\eqdef \left\{\mu\in \R^v: T_{1,n}(\mu)\le c_{n,1-\alpha}^{*}\right\},
\]
where $c_{n,\alpha}^{*}(\omega)\eqdef\inf\{x:\FD{1,n}{*}(\omega,x)\ge \alpha\}$ is the conditional $\alpha$-quantile of $\FD{1,n}{*}$. In practice, if the exact distribution of $T_{j,n}^{*}$, $j=1,2$ is not available, it can be replaced with a suitable Monte-Carlo estimator.

\medskip

\subsection{Block bootstrap (BB).}
First, we suggest a variant of the block bootstrap, which is extensively studied in the time-series and spatial literature. Specifically, we choose the maximal block radius $s_n>0$ and define $n$ overlapping blocks $\{B_{n,1},\ldots,B_{n,n}\}$ with $B_{n,k}\eqdef N_n(k;s_n+1)$. That is $B_{n,k}$ is an $(s_n+1)$ open neighborhood of the node $k$. Then we randomly select $K_n\eqdef \floor{n/\delta_n(s_n)}$ blocks $\{B_{n,1}^{*},\ldots,B_{n,K_n}^{*}\}$ with \textit{replacement} (note that $\delta_n(s_n)$ is the average block size) which yields a bootstrap sample
\[
	Y_n^{*}=\big\{Y_{n,B_{n,k}^{*}}:1\le k\le K_n\big\}.
\]
Formally, let $\{u_1,\ldots,u_{K_n}\}$ be i.i.d.\ $U\{1,n \}$ random variables defined on $(\Omega,\H,\PM)$ and independent of $\G_n$. Then the $k$-th resampled block is defined as $B_{n,k}^{*}=B_{n,u_k}$ and, therefore, for $1\le k\le K_n$ and $1\le l\le n$, $\PR{B_{n,k}^{*}=B_{n,l}\mid \G_n}=n^{-1}$ a.s. For the ease of exposition we assume that $n/\delta_n(s_n)$ is an integer.

The size of the bootstrap sample $L_n\eqdef\sum_{k=1}^{K_n}\abs{B_{n,k}^{*}}$ is random conditional on the data and depends on the distribution of $\abs{N_n(\csdot;s_n)}$ given the network $G_n$. However, on average it is expected to be close to $n$, (in fact, the conditional expectation of $L_n$ given $\CS$ is exactly $n$). Also in the time series case this approach reduces to a variant of the moving blocks bootstrap with unequally sized blocks such that blocks located near the endpoints have smaller size.

Let $Z_{n,k}^{*}\eqdef\sum_{j\in B_{n,k}^{*}}Y_{n,j}$ and let $\tilde{Y}_n^{*}\eqdef n^{-1}\sum_{k=1}^{K_n}Z_{n,k}^{*}$ be the quasi-average of the bootstrap sample $Y_n^{*}$, which replaces the sample average in the bootstrap versions of $T_{1,n}$ and $T_{2,n}$. We could also consider the true average of a pseudo-sample, i.e., $\bar{Y}_n^{*}\eqdef L_n^{-1}\sum_{k=1}^{K_n}Z_{n,k}^{*}$. However, $L_n$ is not independent of the blocks sums and, as mentioned before, its distribution depends on the underlying network topology. As a result, it is relatively difficult to find a ``smooth'' approximation of the distribution of $\sqrt{L_n}\bar{Y}_n^{*}$ which guarantees the first-order consistency of the bootstrap (in particular, the suggested resampling scheme may not be appropriate in this case). In addition, since the conditional expectation of $\tilde{Y}_n$ given $\G_n$ differs from the sample average, we replace the true parameter $\mu_n$ with $\mu_n^{*}\eqdef \E[\tilde{Y}_n^{*}\mid \G_n]$. As indicated in \cite{Lahiri:92} in the time-series context, replacing $\mu_n$ with $\bar{Y}_n$ introduces an additional bias which does not allow for second-order improvements over the normal approximation \citep[see also][Section~2.7.1]{Lahiri:03:Resampling}. The BB counterparts of the test statistics in \eqref{eq:test_stats} are given by
\begin{align*}
	T_{1,n}^{*}&=\sqrt{n}\normin{\tilde{Y}_n^{*}-\mu_n^{*}}, \qtext{and} \\
	T_{2,n}^{*}&=\sqrt{n}\big(\phi(\tilde{Y}_n^{*})-\phi(\mu_n^{*})\big).
\end{align*}

The conditional variance of the scaled sample mean $\Sigma_n$ can be estimated using the bootstrap version $\Sigma_n^{*}\equiv \Var(\sqrt{n}\tilde{Y}_n^{*}\mid \G_n)$. Since $\{Z_{n,1}^{*},\ldots,Z_{n,K_n}^{*}\}$ are conditionally independent given $\G_n$,
\[
	\Sigma_n^{*}=\frac{1}{\delta_n(s_n)}\left(\frac{1}{n}\sum_{i\in N_n}Z_{n,i}Z_{n,i}^{\top}-\bar{Z}_n\bar{Z}_n^{\top}\right) \qtext{a.s.},
\]
where $Z_{n,i}\eqdef \sum_{j\in B_{n,j}}Y_{n,j}$ and $\bar{Z}_n\eqdef n^{-1}\sum_{i\in N_n}Z_{n,i}$. By construction the matrix $\Sigma_n^{*}$ is positive semidefinite and its form is similar to the network HAC estimator \eqref{eq:HAC}. To see this let
\begin{equation}
\label{eq:weight_fun}
	\omega_n(i,j)\eqdef \frac{\abs{N_n(i;s_n+1)\cap N_n(j;s_n+1)}}{\delta_n(s_n)}
\end{equation}
(when $i=j$ we denote this quantity by $\omega_n(i)$). Then
\[
	\Sigma_n^{*}=\frac{1}{n}\sum_{i,j\in N_n}\omega_n(i,j)(Y_{n,i}-\mu_n)(Y_{n,j}-\mu_n)^{\top}+R_n \qtext{a.s.},
\]
where $\E[\norm{R_n}_F\mid \CS]\to 0$ a.s.\ under some conditions, given later. It is worth mentioning that when $\mu_n=0$ a.s., the remainder term $R_n=0$ a.s. Unlike a typical kernel, the weighting functions $\omega_n(\csdot,\csdot)$ depends on the network topology and it is not bounded by $1$. However, for fixed $i\in N_n$ it is decreasing in the distance between $i$ and $j$. Let $\tilde{\omega}\eqdef \sup_{n}\max_{i\ne j}\omega_n(i,j)$, $\tilde{\mu}_p\eqdef\sup_{n,i\in N_n}\norm{Y_{n,i}}_{\CS,p}$ for $p>0$, and
\begin{equation}
\label{eq:graph_hom}
	\Delta_n(s;k)\eqdef \frac{1}{n}\sum_{i\in N_n}\abs{\abs{N_n(i;s+1)}-\delta_n(s)}^k,
\end{equation}
which is the $k$-th absolute central moment of the sizes of the $(s+1)$-neighborhoods. The following assumptions provide sufficient conditions for the consistency of $\Sigma_n^{*}$.

\begin{assumption}
\label{assu:BB1}
The sequence $\seq{(G_n,s_n)}$ is such that w.p.1 $\tilde{\omega}<\infty$ and
\begin{itemize}[leftmargin=1.65em]
	\setlength\itemsep{0.3em}
	\item[(a)] $\Delta_n(s_n;2)/\delta_n(s_n)+D_n(s_n)/\sqrt{\delta_n(s_n)n}\to 0$,
	\item[(b)] $\max_{i\in N_n}\abs{\sum_{j\in B_{n,i}}(\omega_n(j)-1)}/\sqrt{n}\to 0$,
	\item[(c)] $n^{-1}\sum_{i\in N_n}\sum_{j\in N_n^{\partial}(i;s)}\abs{\omega_n(i,j)-1}\gamma_{n,s}\to 0$ for all $s\ge 1$,
\end{itemize}
\end{assumption}

Assumption \ref{assu:BB1} imposes restrictions on admissible networks topologies. Specifically, the consistency of $\Sigma_n^{*}$ requires a certain degree of homogeneity of the resampled blocks which is characterized by various moments of the weights $\seq{\omega_n(i,j):i,j\in N_n}$. For example, condition $(a)$ requires that the sample variance of $\{\abs{B_{n,i}}\}$ increases at a lower rate than the average block size. It also guarantees that $\mu_n^{*}$ is a consistent estimator of the mean $\mu_n$ and that for large samples the size of a pseudo-sample, $L_n$ is close to $n$. In fact, $\E[\abs{L_n/n-1}\mid \CS]\to 0$ a.s.\ because
\begin{align*}
	\E[\abs{L_n/n-1}\mid \CS] &=\frac{1}{n}\E\left[\abs{\sum_{k=1}^{K_n}\left(\abs{B_{n,k}^{*}}-\delta_n(s_n)\right)}\mid \CS\right] \\
	&\le \frac{1}{n}\sum_{k=1}^{K_n} \Delta_n(s_n;1)\le \frac{\sqrt{\Delta_n(s_n;2)}}{\delta_n(s_n)} \qtext{a.s.}
\end{align*}
This condition is clearly satisfied in the time series context when $s_n=o(\sqrt{n})$ (although, it has been shown that the consistency of the moving block bootstrap in this case holds for $s_n=o(n)$ \citep[see, e.g.,][]{Calhoun:18}). However, it does not hold for unweighted ``star'' networks and $s_n\equiv 1$ because $\Delta_n(1;2)\ge [\Delta_n(1;1)]^2\to 4$ and $\delta_n(1)\to 3$ as $n\to \infty$. In practice, one can compute $\Delta_n(s_n;2)$ for a given graph to see whether this quantity is small relative to the average block size.

Condition (c) ensures that all the non-zero autocovariances are estimated consistently. It is similar to an assumption on kernel functions used in HAC estimation, that is in the limit the value of a kernel at each $s$ must converge to 1. In addition, if $\tilde{\gamma}_s>0$ with positive probability for all $s\ge 1$, then the parameter $s_n$ must go to infinity for this condition to hold.

\begin{assumption}
\label{assu:BB2}
There exists $r>2$ such that w.p.1 $\tilde{\mu}_{2r}<\infty$ and
\begin{itemize}[leftmargin=1.65em]
	\setlength\itemsep{0.3em}
	\item[(a)] $\limsup_{n\to\infty}\sum_{s\ge 1}\delta_n^{\partial}(s)\gamma_{n,s}^{1-\frac{2}{r}}<\infty$,
	\item[(b)] $n^{-2}\sum_{s\ge 0}\abs{H_n(s,2s_n+1)}\gamma_{n,s}^{1-\frac{2}{r}}\to 0$.
\end{itemize}
\end{assumption}

The conditions of Assumption \ref{assu:BB2} are similar to those needed for the consistency of the network HAC estimator \eqref{eq:HAC}. In particular, condition (b) gives a rule of thumb for the choice of the truncation parameter $s_n$ (see \citealiasp{Kojevnikov/Marmer/Song:20}, Section~4.1). Also condition (a) implies that the elements of the true variance $\Sigma_n$ do not diverge to $\pm\infty$. To see this note that for $1\le k,l\le v$ and some constant $C>0$,
\[
	\abs{[\Sigma_n]_{kl}}\le C(\tilde{\mu}_{2r}^2\vee 1)\left(1+\sum_{s\ge 1}\delta_n^{\partial}(s)\gamma_{n,s}^{1-\frac{2}{r}}\right) \qtext{a.s.}
\]
Therefore, $\limsup_{n\ge 1}\abs{[\Sigma_n]_{kl}}<\infty$ a.s.

\begin{prop}
\label{prop:BB_var}
Suppose that Assumptions \ref{assu:BB1} and \ref{assu:BB2} hold. Then
\[
	\E[\norm{\Sigma_n^{*}-\Sigma_n}_F\mid \CS]\to 0 \qtext{a.s.}
\]
\end{prop}

The result of Proposition \ref{prop:BB_var} implies that $\Sigma_n^{*}$ is a consistent estimator of $\Sigma_n$. Therefore, assuming that $\Sigma_n\to \Sigma$ a.s.\ and $\sqrt{n}(\bar{Y}_n-\mu_n)$ converges $\CS$-weakly to a conditionally normal random vector with variance $\Sigma$, we may use Corollaries \ref{corr:bootstrap_consistency1} and \ref{corr:bootstrap_consistency2} to establish the consistency of the bootstrap distributions. For example, one may employ Theorem 3.2 in \citealias{Kojevnikov/Marmer/Song:20} together with the Cram\'er–Wold device and Lemma \ref{lemma:aux_weak_eq}.

\begin{assumption}
\label{assu:BB3}
$\Sigma_n$ converges a.s.\ to a $\CS$-measurable, positive definite matrix $\Sigma$, and
\[
	\sqrt{n}(\bar{Y}_n-\mu_n)\to \Sigma^{1/2}\eta \qtext{$\CS$-weakly},
\]
where $\eta\sim\ND{0}{I_v}$ independent of $\CS$.\footnote{
	Assumption \ref{assu:BB3} is made merely for ease of exposition. In view of Theorem \ref{thm:bootstrap_consistency1} it can be omitted at the expense of establishing additional Berry-Essen type bounds.
}
\end{assumption}

In addition, we introduce the local versions of some measures of the network denseness. Specifically, for $s,m\ge 0$ we define
\begin{align*}
	\delta_{loc,n}^{\partial}(s,m)&\eqdef \max_{i\in N_n}\sum_{j\in N_n(i;m)}\frac{\abs{N_n^{\partial}(j;s)\cap N_n(i;m)}}{\abs{N_n(i;m)}}
\intertext{and}
	h_{loc,n}(s,m)&\eqdef \max_{i\in N_n}\frac{\abs{H_n(s,\infty)\cap N_n^4(i;m)}}{\abs{N_n(i;m)}^3}.
\end{align*}
These measures are constructed in a way such that for any $m\ge 0$, $\delta_{loc,n}^{\partial}(0,m)=h_{loc,n}(0,m)=1$. Also note that $h_{loc,n}(s,m)\le \delta_{loc,n}^{\partial}(s,m)$.

\begin{assumption}
\label{assu:BB4}
There exists $p>2$ such that w.p.1 $\tilde{\mu}_{2p}<\infty$ and
\begin{align*}
	&\left(\delta_n(s_n)/n\right)^{1/3}\sum_{s\ge 0}\delta_{loc,n}^{\partial}(s,s_n)\gamma_{n,s}^{1-\frac{2}{p}} \\
	&\qquad+\left(\delta_n^{5/2}(s_n)/n\right)^{2/3}\sum_{s\ge 0}h_{loc,n}(s,s_n)\gamma_{n,s}^{1-\frac{2}{p}}\to 0.
\end{align*}
\end{assumption}

When the following summability condition holds:
\[
	\limsup_{n\to\infty}\sum_{s\ge 1}\delta_{loc,n}^{\partial}(s,s_n)\gamma_{n,s}^{1-\frac{2}{p}}<\infty \qtext{a.s.},
\]
Assumption \ref{assu:BB3} reduces to $\delta_n^{5/2}(s_n)/n\to 0$ a.s. In particular, if for each $n\ge 1$ the blocks $\seq{B_{n,k}}$ have the same size $l_n<n$, it suffices to assume that the weak dependence coefficients raised to the power $1-2/p$ are a.s.\ summable and $l_n=o(n^{2/5})$. Note that this assumption also explicitly requires $K_n\to\infty$ as $n\to\infty$.

\begin{prop}
\label{prop:BB_consistency}
Suppose that Assumptions \ref{assu:BB1}-\ref{assu:BB4} hold. Then $\FD{1,n}{*}$ is conditionally $\DKM$-consistent given $\CS$. If, in addition, $\mu_n$ converges a.s.\ to a $\CS$-measurable random vector $\mu$ and $\nabla\phi(\mu)\ne 0$ a.s., then $\FD{2,n}{*}$ is conditionally $\DKM$-consistent given $\CS$.
\end{prop}

\medskip

\subsection{Dependent wild bootstrap (DWB).}
The dependent wild bootstrap for time-series was introduced in \cite{Shao:10}. This method approximates the finite-sample distribution of $T_n$ by mimicking the autocovariance structure of the underlying sample. In particular, adapting to our framework, assume that $\CS=\{\emptyset,\Omega\}$ and let $G_n$ be an unweighted ``line'' network. Consider an $n$-dimensional, zero mean random vector $W_n$ defined on $(\Omega,\H,\PM)$ and independent of $Y_n$ such that $\Var(W_{n,i})=1$ and $\Cov(W_{n,i},W_{n,j})=\kappa(d_n(i,j)/(s_n+1))$, where $\kappa(\csdot)$ is a positive-definite kernel function and $s_n$ is a bandwidth parameter. The DWB pseudo-sample $Y_n^{*}$ is defined as follows:
\[
	Y_{n,i}^{*}=\bar{Y}_n+(Y_{n,i}-\bar{Y}_n)W_{n,i},\quad i\in N_n.
\]

Let $\bar{Y}_n^{*}\eqdef n^{-1}\sum_{i\in N_n}Y_{n,i}^{*}$. By construction, $\E[\bar{Y}_n^{*}\mid \G_n]=\bar{Y}_n$ so that in contrast to the block bootstrap, the statistic $\sqrt{n}(\bar{Y}_n^{*}-\bar{Y}_n)$ is unbiased given $\G_n$. In addition, noticing that $\kappa(0)=1$, the conditional variance of the scaled bootstrap mean given $\G_n$ is
\begin{align*}
	\Sigma_n^{*}&=\frac{1}{n}\sum_{i,j\in N_n}\Cov(W_{n,i},W_{n,j})(Y_{n,i}-\bar{Y}_n)(Y_{n,j}-\bar{Y}_n)^{\top} \\
	&=\frac{1}{n}\sum_{i,j\in N_n}\kappa\left(\frac{d_n(i,j)}{s_n+1}\right)(Y_{n,i}-\bar{Y}_n)(Y_{n,j}-\bar{Y}_n)^{\top},
\end{align*}
which is a version of the network HAC estimator \eqref{eq:HAC}. Then under certain regularity conditions the DWB is first-order consistent for smooth functions of the mean.

For general graphs, however, positive definiteness of the kernel function $\kappa$ does not imply that the matrix $[\kappa(d_n(i,j)/(s_n+1))]_{i,j\in N_n}$ is positive semi-definite (see \citealiasp{Kojevnikov/Marmer/Song:20}, Section 4.1). Therefore, in general, we cannot guarantee the existence of a random vector with the required covariance structure. A simple way to overcome this issue is to rely on the topology of a given network. Consider the matrix $\Omega_n=[\omega_n(i,j)]_{i,j\in N_n}$, where $\omega_n$ is defined in \eqref{eq:weight_fun}.

\begin{claim}
$\Omega_n$ is positive semi-definite.
\end{claim}
\begin{proof}
Let $c\in \R^n$ and $\xi_i\eqdef \sum_{j\in N_n(i;s_n+1)}c_j$. Then, since $(j,k)\in N_n^2(i;s_n+1)$ if and only if $i\in N_n(j;s_n+1)\cap N_n(k;s_n+1)$,
\[
	\sum_{i\in N_n}\xi_i^2=\sum_{i\in N_n}\sum_{j,k\in N_n(i;s_n+1)}c_jc_k=\sum_{i,j\in N_n}c_ic_j\omega_n(i,j)\delta_n(s_n),
\]
Therefore,
\[
	c^{\top}\Omega_n c=\sum_{i,j\in N_n}c_ic_j \omega_n(i,j)=\sum_{i\in N_n}\xi_i^2/\delta_n(s_n)\ge 0. \qedhere
\]
\end{proof}

Consequently, we consider a random vector $W_n$ satisfying the following assumption.

\begin{assumption}
\label{assu:DWB1}
$W_n$ is conditionally independent of $Y_n$ given $\CS$ with $\E[W_n\mid \CS]=0$ a.s and $\E[W_n W_n^{\top}\mid \CS]=\Omega_n$ a.s.
\end{assumption}

Under Assumption \ref{assu:DWB1} the bootstrap variance estimator given by
\begin{equation}
\label{eq:DWB_var}
	\Sigma_n^{*}=\frac{1}{n}\sum_{i,j\in N_n}\omega_n(i,j)(Y_{n,i}-\bar{Y}_n)(Y_{n,j}-\bar{Y}_n)^{\top} \qtext{a.s.}
\end{equation}
is positive semi-definite. We impose the next conditions on the sequence of networks, which in combination with Assumption \ref{assu:BB2}, ensure the consistency of $\Sigma_n^{*}$.

\begin{assumption}
\label{assu:DWB2}
The sequence $\seq{(G_n,s_n)}$ is such that w.p.1 $\tilde{\omega}<\infty$ and
\begin{itemize}[leftmargin=1.65em]
	\setlength\itemsep{0.3em}
	\item[(a)] $\Delta_n(s_n;1)/\delta_n(s_n)+D_n(s_n)/n\to 0$,
	\item[(b)] $n^{-1}\sum_{i\in N_n}\sum_{j\in N_n^{\partial}(i;s)}\abs{\omega_n(i,j)-1}\gamma_{n,s}\to 0$ for all $s\ge 1$.
\end{itemize}
\end{assumption}

The conditions given in Assumption \ref{assu:DWB1} are clearly weaker than those needed for the consistency of the BB variance estimator, which follows from the fact that $\Delta_n(s_n;1)\le \sqrt{\Delta_n(s_n;2)}$ and $\delta_n(s_n)\le n$. Therefore, the DWB estimator \eqref{eq:DWB_var} is likely to be consistent for a wider class of networks. As in the case of the block bootstrap we assume that the true variance $\Sigma_n$ converges a.s.\ to a $\CS$-measurable matrix $\Sigma$.

\begin{prop}
\label{prop:DWB_var}
Suppose that Assumptions \ref{assu:DWB2} and \ref{assu:BB2} hold. Then
\[
	\E[\norm{\Sigma_n^{*}-\Sigma_n}_F\mid \CS] \to 0 \qtext{a.s.}
\]
\end{prop}

First, we consider the Gaussian case. That is, we take $W_n=\Omega_n^{1/2}\zeta_n$, where $\zeta_n$ is the standard normal random vector in $\R^n$ independent of $\G_n$. From the practical perspective it is a convenient choice, especially when $n$ is large because a sample from a multivariate normal distribution can be easily generated. Moreover, efficient algorithms for finding the square root of positive semidefinite matrices are available. We refer to \cite{Higham:08:FM} for details. As noted in \cite{Shao:10}, although the DWB sample with Gaussian weights may not match non-zero higher-order cumulants of the original process, it is difficult to choose the joint distribution of $W_n$ that fits those cumulants, and performance of the DWB primarily depends on the choice of the truncation parameter $s_n$.

In this case, conditionally on $\G_n$, the statistic
\[
	\sqrt{n}(\bar{Y}_n^{*}-\bar{Y}_n)=\frac{1}{\sqrt{n}}\sum_{i\in N_n}W_{n,i}(Y_{n,i}-\bar{Y}_n)
\]
is also normal with zero mean and variance given in \eqref{eq:DWB_var}. Therefore, the conditional distribution of the DWB counterpart of the test statistic $T_{1,n}$,
\[
	T_{1,n}^{*}=\sqrt{n}\norm{\bar{Y_n}^{*}-\bar{Y}_n}
\]
given $\G_n$ is known and is the same as the conditional distribution of the asymptotic Gaussian approximation $\normin{\Sigma_n^{*1/2} \eta}$, where $\eta$ is as $v$-dimensional standard normal random vector independent of $\G_n$, and the latter converges $\CS$-weakly to $\normin{\Sigma^{1/2}\eta}$ by Lemma \ref{lemma:aux_weak_conv}. A more interesting case, however, arises when considering the second test statistic $T_{2,n}$ because for nonlinear transformations the conditional distribution of its bootstrap analog,
\[
	T_{2,n}^{*}=\sqrt{n}\left(\phi(\bar{Y_n}^{*})-\phi(\bar{Y}_n)\right),
\]
is typically unavailable. Then in the Gaussian case the DWB is consistent without any further restriction on the topology of the sequence of networks $\seq{G_n}$. We only need to assume that $\sqrt{n}(\bar{Y}_n-\mu_n)$ converges $\CS$-weakly to a conditionally normal random vector and the asymptotic variance of $T_{2,n}$ is a.s.\ positive.

\begin{prop}
\label{prop:DWB_consistency1}
Suppose that $W_n$ is Gaussian, Assumptions \ref{assu:DWB1}, \ref{assu:DWB2}, \ref{assu:BB2}, and \ref{assu:BB3} hold. Then $\FD{1,n}{*}$ is conditionally $\DKM$-consistent given $\CS$. If, in addition, $\mu_n$ converges a.s.\ to a $\CS$-measurable random vector $\mu$ and $\nabla\phi(\mu)\ne 0$ a.s., then $\FD{2,n}{*}$ is conditionally $\DKM$-consistent given $\CS$.
\end{prop}

Given another choice of $W_n$, the process $\arr{\xi_{n,i}\eqdef W_{n,i}(Y_{n,i}-\bar{Y}_n):i\in N_n}$ is $s_n$-dependent conditionally on $\G_n$, i.e., $\xi_{n,i}$ and $\xi_{n,j}$ are conditionally independent given $\G_n$ whenever $j\notin B_{n,i}\eqdef N_n(i;s_n+1)$. Consequently, in addition to the assumptions of Proposition \ref{prop:DWB_consistency1}, we need to control the behavior of the third conditional moments of $W_n$ and the neighborhoods $\seq{B_{n,i}}$ such that the bootstrap distributions $\FD{1,n}{*}$ and $\FD{2,n}{*}$ in this case approach ones under the Gaussian weights as $n\to \infty$.

\begin{prop}
\label{prop:DWB_consistency2}
Suppose that Assumptions \ref{assu:DWB1}, \ref{assu:DWB2}, \ref{assu:BB2}, and \ref{assu:BB3} hold, and
\begin{equation}
\label{eq:DWB_consistency2}
	\frac{1}{n^{3/2}}\sum_{i\in N_n}\sum_{j\in B_{n,i}}\sum_{k\in B_{n,i}\cup B_{n,j}}\prod_{l\in\{i,j,k\}}\norm{W_{n,l}}_{\CS,3}\to 0 \qtext{a.s.}
\end{equation}
Then $\FD{1,n}{*}$ is conditionally $\DKM$-consistent given $\CS$. If, in addition, $\mu_n$ converges a.s.\ to a $\CS$-measurable random vector $\mu$ and $\nabla\phi(\mu)\ne 0$ a.s., then $\FD{2,n}{*}$ is conditionally $\DKM$-consistent given $\CS$.
\end{prop}

\begin{remark*}
The convergence condition in \eqref{eq:DWB_consistency2} is quite strong. In particular, in a simple case when the neighborhoods $\seq{B_{n,i}}$ have the same size $l_n<n$ for all $n\ge 1$ and $\sup_{n,i\in N_n}\norm{W_{n,i}}_{\CS,3}<\infty$ a.s., it requires $l_n=o(n^{1/4})$. Therefore, it is of high interest to find a better way to handle network dependent processes under $m$-dependence.
\end{remark*}



\section{Conclusion}

Nonparametric bootstrapping for time series and spatial processes has been extensively studied in the past decades. Thus, various resampling methods are now available for statistical analysis of dependent data in these cases. However, the lack of regular structure in networks renders the use of these techniques for bootstrap-based inference in the case of network dependent processes impracticable. In this paper we proposed a block-based method and a variant of the dependent wild bootstrap suitable for the latter processes satisfying the conditional version of \cite{Doukhan/Louhichi:99}'s $\psi$-weak dependence condition. We established the first-order validity of these methods to construct confidence sets for the mean of a network dependent process. In addition, we showed their consistency under the smooth function model conditionally on a common shock of a general form. Finally, the corresponding bootstrap variance estimators can be used for asymptotic inference instead of the network HAC estimator, which is not necessarily positive semi-definite.

As for the future directions, having a continuity theorem and other related results similar to ones established in \cite{Belyaev:00} but under convergence in conditional probability would significantly weaken the bootstrap consistency conditions derived in this paper. In addition, an extension of these methods for bootstrapping $M$-estimators and empirical processes is of great importance for applied research.

\newpage
\bibliographystyle{elsart-harv}
\bibliography{network_bootstrap}

\begin{thebibliography}{45}
\expandafter\ifx\csname natexlab\endcsname\relax\def\natexlab#1{#1}\fi
\expandafter\ifx\csname url\endcsname\relax
  \def\url#1{\texttt{#1}}\fi
\expandafter\ifx\csname urlprefix\endcsname\relax\def\urlprefix{URL }\fi

\bibitem[{Athreya and Lahiri(2006)}]{Athreya:2006:MTP}
Athreya, K.~B., Lahiri, S.~N., 2006. Measure Theory and Probability Theory
  ({Springer} Texts in Statistics). Springer-Verlag, Berlin, Heidelberg.

\bibitem[{Belyaev and Sjöstedt-de Luna(2000)}]{Belyaev:00}
Belyaev, Y., Sjöstedt-de Luna, S., 2000. Weakly approaching sequences of
  random distributions. Journal of Applied Probability 37~(3), 807–--822.

\bibitem[{Bentkus(2003)}]{Bentkus:03}
Bentkus, V., 2003. On the dependence of the {Berry-Esseen} bound on dimension.
  Journal of Statistical Planning and Inference 113, 385--–402.

\bibitem[{Berti et~al.(2006)Berti, Pratelli, and Rigo}]{Berti:06}
Berti, P., Pratelli, L., Rigo, P., 2006. Almost sure weak convergence of random
  probability measures. Stochastics 78~(2), 91--97.

\bibitem[{Bühlmann(1993)}]{Buhlmann:93}
Bühlmann, P.~L., 1993. The blockwise bootstrap in time series and empirical
  processes. Ph.D. thesis, Swiss Federal Institute of Technology Zürich.

\bibitem[{Calhoun(2018)}]{Calhoun:18}
Calhoun, G., 2018. Block bootstrap consistency under weak assumptions.
  Econometric Theory.

\bibitem[{Carlstein(1986)}]{Carlstein:86}
Carlstein, E., 1986. The use of subseries values for estimating the variance of
  a general statistic from a stationary sequence. The Annals of Statistics
  14~(3), 1171--1179.

\bibitem[{Chernozhukov et~al.(2013)Chernozhukov, Chetverikov, and
  Kato}]{Chernozhukov/Chetverikov/Kato:13}
Chernozhukov, V., Chetverikov, D., Kato, K., 2013. Gaussian approximations and
  multiplier bootstrap for maxima of sums of high-dimensional random vectors.
  Annals of Statistics 41~(6), 2786--2819.

\bibitem[{Cohen and Elliott(2015)}]{Cohen:15:StochCalc}
Cohen, S., Elliott, R.~J., 2015. Stochastic Calculus and Applications, 2nd
  Edition. Probability and Its Applications. Birkhäuser Basel.

\bibitem[{Comets and Jan\v{z}ura(1998)}]{Comets/Janzura:98}
Comets, F., Jan\v{z}ura, M., 1998. A central limit theorem for conditionally
  centered random fields with an application to {Markov} fields. Journal of
  Applied Probability 35, 608--621.

\bibitem[{Conley(1999)}]{Conley:99}
Conley, T.~G., 1999. {GMM} estimation with cross-sectional dependence. Journal
  of Econometrics 92, 1--45.

\bibitem[{Crimaldi(2009)}]{Crimaldi:09}
Crimaldi, I., 2009. An almost sure conditional convergence result and an
  application to a generalized {P\'olya} urn. International Mathematical Forum
  4~(23), 1139--1156.

\bibitem[{Dedecker et~al.(2007)Dedecker, Doukhan, and
  Lang}]{Dedecker/Doukhan:07:WeakDep}
Dedecker, J., Doukhan, P., Lang, G., 2007. Weak Dependence: With Examples and
  Applications. Lecture Notes in Statistics. Springer, New York.

\bibitem[{Doukhan and Louhichi(1999)}]{Doukhan/Louhichi:99}
Doukhan, P., Louhichi, S., 1999. A new weak dependence condition and
  applications to moment inequalities. Stochastic Processes and their
  Applications 84~(2), 313--342.

\bibitem[{Durrett(2010)}]{Durrett:10:Prob}
Durrett, R., 2010. Probability: Theory and Examples, 4th Edition. Cambridge
  University Press.

\bibitem[{Embrechts and Hofert(2013)}]{Embrechts/Hofert:13}
Embrechts, P., Hofert, M., 2013. A note on generalized inverses. Mathematical
  Methods of Operations Research 77~(3), 423--432.

\bibitem[{Garcia-Soidan et~al.(2014)Garcia-Soidan, Menezes, and
  Rubinos}]{Garcia:13}
Garcia-Soidan, P., Menezes, R., Rubinos, O., 2014. Bootstrap approaches for
  spatial data. Stochastic Environmental Research and Risk Assessment 28~(5),
  1207--1219.

\bibitem[{Gonçalves and Politis(2011)}]{Goncalves:11}
Gonçalves, S., Politis, D., 2011. Discussion: Bootstrap methods for dependent
  data: A review. Journal of the Korean Statistical Society 40~(4), 383--386.

\bibitem[{Gonçalves and White(2002)}]{Goncalves/White:02}
Gonçalves, S., White, H., 2002. The bootstrap of the mean for dependent
  heterogeneous arrays. Econometric Theory 18~(6), 1367--1384.

\bibitem[{Hall(1992)}]{Hall:92:Bootstrap}
Hall, P., 1992. The Bootstrap and Edgeworth Expansion. Springer Series in
  Statistics. Springer-Verlag, New York, Berlin, Paris.

\bibitem[{Higham(1988)}]{Higham:88}
Higham, N., 1988. Computing a nearest symmetric positive semidefinite matrix.
  Linear Algebra and its Applications 103, 103--118.

\bibitem[{Higham(2002)}]{Higham:02}
Higham, N., 2002. Computing the nearest correlation matrix - a problem from
  finance. IMA Journal of Numerical Analysis 22~(3), 329--343.

\bibitem[{Higham(2008)}]{Higham:08:FM}
Higham, N.~J., 2008. Functions of Matrices: {Theory} and Computation. Society
  for Industrial and Applied Mathematics, Philadelphia, PA, USA.

\bibitem[{Jenish and Prucha(2009)}]{Jenish/Prucha:09}
Jenish, N., Prucha, I.~R., 2009. Central limit theorems and uniform laws of
  large numbers for arrays of random fields. J. Econometrics 150~(1), 86--98.

\bibitem[{Kallenberg(2002)}]{Kallenberg:02:FMP}
Kallenberg, O., 2002. Foundations of Modern Probability, 2nd Edition.
  Probability and its Applications (New York). Springer-Verlag, New York.

\bibitem[{Kelejian and Prucha(2007)}]{Kelejian/Prucha:07}
Kelejian, H.~H., Prucha, I.~R., 2007. {HAC} estimation in a spatial framework.
  Journal of Econometrics, 131--154.

\bibitem[{Kelejian and Prucha(2010)}]{Kelejan/Prucha:07_2}
Kelejian, H.~H., Prucha, I.~R., 2010. Specification and estimation of spatial
  autoregressive models with autoregressive and heteroskedastic disturbances.
  Journal of Econometrics 157~(1), 53--67, nonlinear and Nonparametric Methods
  in Econometrics.

\bibitem[{Künsch(1989)}]{Kunsch:89}
Künsch, H.~R., 1989. The jackknife and the bootstrap for general stationary
  observations. The Annals of Statistics 17, 1217--1241.

\bibitem[{Kojevnikov et~al.(2020)Kojevnikov, Marmer, and
  Song}]{Kojevnikov/Marmer/Song:20}
Kojevnikov, D., Marmer, V., Song, K., 2020. Limit theorems for network
  dependent random variables. Journal of Econometrics.
\newline\urlprefix\url{https://doi.org/10.1016/j.jeconom.2020.05.019}

\bibitem[{Lahiri(1992)}]{Lahiri:92}
Lahiri, S.~N., 1992. Edgeworth correction by 'moving block' bootstrap for
  stationary and nonstationary data. In: LePage, R., Billard, L. (Eds.),
  Exploring the limits of bootstrap. John Wiley \& Sons, New York; Chichester,
  pp. 183--214.

\bibitem[{Lahiri(2003)}]{Lahiri:03:Resampling}
Lahiri, S.~N., 2003. Resampling Methods for Dependent Data. Springer Series in
  Statistics.

\bibitem[{Liu and Singh(1992)}]{Liu/Singh:92}
Liu, R.~Y., Singh, K., 1992. Moving blocks jackknife and bootstrap capture weak
  dependence. In: LePage, R., Billard, L. (Eds.), Exploring the limits of
  bootstrap. John Wiley \& Sons, New York; Chichester, pp. 225--248.

\bibitem[{Matyas(1999)}]{Matyas:99:GMM}
Matyas, L., 1999. Generalized Method of Moments Estimation. Cambridge
  University Press, Cambridge.

\bibitem[{Paparoditis and Politis(2001)}]{Paparoditis/Politis:01}
Paparoditis, E., Politis, D.~N., 2001. Tapered block bootstrap. Biometrika
  88~(4), 1105--1119.

\bibitem[{Politis and Romano(1992)}]{Politis/Romano:92}
Politis, D.~N., Romano, J.~P., 1992. A circular block-resampling procedure for
  stationary data. In: LePage, R., Billard, L. (Eds.), Exploring the limits of
  bootstrap. John Wiley \& Sons, New York; Chichester, pp. 263--270.

\bibitem[{Politis and Romano(1994)}]{Politis/Romano:94}
Politis, D.~N., Romano, J.~P., 1994. The stationary bootstrap. Journal of the
  American Statistical Association 89~(428), 1303--1313.

\bibitem[{Politis et~al.(1999)Politis, Romano, and Wolf}]{Politis:99}
Politis, D.~N., Romano, J.~P., Wolf, M., 1999. Subsampling. Springer.

\bibitem[{Rhee and Talagrand(1986)}]{Rhee/Talagrand:86}
Rhee, W., Talagrand, M., 1986. Uniform bound in the central limit theorem for
  {Banach} space valued dependent random variables. Journal of Multivariate
  Analysis 20~(2), 303--320.

\bibitem[{Röllin(2013)}]{Rollin:13}
Röllin, A., 2013. Stein's method in high dimensions with applications. Ann.
  Inst. H. Poincaré Probab. Statist. 49~(2), 529--549.

\bibitem[{Sengupta et~al.(2015)Sengupta, Shao, and Wang}]{Sengupta/Shao:15}
Sengupta, S., Shao, X., Wang, Y., 2015. The dependent random weighting. Journal
  of Time Series Analysis 36~(3), 315--326.

\bibitem[{Shao and Tu(1995)}]{Shao:95:Bootstrap}
Shao, J., Tu, D., 1995. The Jackknife and Bootstrap. Springer-Verlag, Berlin;
  New York.

\bibitem[{Shao(2010)}]{Shao:10}
Shao, X., 2010. The dependent wild bootstrap. Journal of the American
  Statistical Association 105~(489), 218--235.

\bibitem[{Shiryaev(2016)}]{Shiryaev:16:Probability}
Shiryaev, A.~N., 2016. Probability-1, 3rd Edition. Graduate Texts in
  Mathematics. Springer-Verlag New York.

\bibitem[{Talagrand(2011)}]{Talagrand:11:SpinGlasses}
Talagrand, M., 2011. Mean Field Models for Spin Glasses, Volume I: Basic
  Examples. Vol.~54 of A Series of Modern Surveys in Mathematics.
  Springer-Verlag.

\bibitem[{Zhang(2011)}]{Zhang:11:MatrixTheory}
Zhang, F., 2011. Matrix Theory: Basic Results and Techniques, 2nd Edition.
  Universitext. Springer-Verlag New York.

\end{thebibliography}

\newpage
\appendix

\section{Proofs of the Main Results}
\label{sec:app_main}

In the following let $\varphi_K$ with $K\in\R_{+}$ denote the element-wise censoring function, i.e., for an indexed family of real numbers $\mathbf{x}\equiv (x_i)_{i\in I}$,
\[
	[\varphi_K(\mathbf{x})]_i\eqdef (-K)\vee(K\wedge x_i), \quad i\in I.
\]

\begin{proof}[\textbf{Proof of Proporsition \ref{prop:local_lipshitz}}]
Fix $\kappa\ge 1$ and let $\xi\eqdef(f\circ h)(Z_{n,A})$ and $\zeta\eqdef(g\circ h)(Z_{n,B})$, where $f,g\in \L_w$ and $(A,B)\in \PS_n(a,b;s)$. Define the censored versions
\[
	\xi_{\kappa}\eqdef(f\circ h\circ \varphi_{\kappa})(Z_{n,A}) \qtext{and}\quad \zeta_{\kappa}\eqdef(g\circ h\circ \varphi_{\kappa})(Z_{n,B}).
\]
Then
\begin{align*}
	\abs{\Cov(\xi,\zeta\mid \CS)}&\le \abs{\Cov(\xi-\xi_{\kappa},\zeta-\zeta_{\kappa}\mid \CS)}+\abs{\Cov(\xi_{\kappa},\zeta_{\kappa}\mid \CS)} \\
	&\le 2\norm{f}_{\infty}\E[\abs{\zeta-\zeta_{\kappa}}\mid \CS]+2\norm{g}_{\infty}\E[\abs{\xi-\xi_{\kappa}}\mid \CS] \\
	&\quad +\abs{\Cov(\xi_{\kappa},\zeta_{\kappa}\mid \CS)} \qtext{a.s.}
\end{align*}
First, $\Lip(f\circ h\circ \varphi_{\kappa})\le 2\eta\kappa^{\tau-1}\Lip(f)$ and $\Lip(g\circ h\circ \varphi_{\kappa})\le 2\eta\kappa^{\tau-1}\Lip(g)$. Therefore,
\begin{align}
	&\begin{aligned}
	\label{eq:local_lipshitz1}
		\abs{\Cov(\xi_{\kappa},\zeta_{\kappa}\mid \CS)}&\le \big(c_1\norm{f}_{\infty}\norm{g}_{\infty} \\
		&\quad +2\eta\kappa^{\tau-1}\{c_2\Lip(f)\norm{g}_{\infty}+c_3\norm{f}_{\infty}\Lip(g)\} \\
		&\quad +(2\eta\kappa^{\tau-1})^2 c_4\Lip(f)\Lip(g)\big)\gamma_{n,s} \qtext{a.s.}
	\end{aligned}
\intertext{Second,}
	&\begin{aligned}
		\label{eq:local_lipshitz2}
		\E[\abs{\xi-\xi_{\kappa}}\mid \CS]&\le \Lip(f)\sum_{i\in A}\E[\norm{h(Z_{n,i})-(h\circ \varphi_{\kappa})(Z_{n,i})}\mid \CS] \\
		&\le C_v\Lip(f)\sum_{i\in A}\E[\norm{Z_{n,i}}_{\infty}^{\tau}\ind\{\norm{Z_{n,i}}_{\infty}>\kappa\}\mid \CS] \\
		&\le C_v\Lip(f)aL\kappa^{\tau-p} \qtext{a.s.},
	\end{aligned}
\end{align}
where $C_v>0$ is a constant depending on $v$ and $\eta$. Similarly,
\begin{align}
\label{eq:local_lipshitz3}
	\E[\abs{\zeta-\zeta_{\kappa}}\mid \CS]&\le C_v\Lip(g)bL\kappa^{\tau-p} \qtext{a.s.}
\end{align}

Since inequalities \eqref{eq:local_lipshitz1}-\eqref{eq:local_lipshitz3} hold for all $\kappa\ge 1$ a.s., they also hold for random $\kappa$ on $\{\kappa\in [1,\infty)\}$. The result follows by setting $\kappa=(\gamma_{n,s}\wedge 1)^{1/(1-p)}$, if $c_4=0$ and $\kappa=(\gamma_{n,s}\wedge 1)^{1/(2-p-\tau)}$, otherwise, and, noticing that $\Cov(\xi,\zeta\mid \CS)=0$ a.s.\ on $\{\gamma_{n,s}=0\}$.
\end{proof}

\medskip

\begin{proof}[\textbf{Proof of Theorem \ref{thm:CWLLN}}]
First, it suffices to show that
\[
	\norm{c^{\top}(\bar{Y}_n-\E[\bar{Y}_n\mid \CS])}_{\CS,1}\to 0 \qtext{a.s.}
\]
for any $c\in \R^v$ with $\norm{c}=1$. Then the proof is similar to one given in the unconditional case. Specifically, for $k>0$, let $\xi_{n,i}^{(k)}\eqdef \varphi_k(c^{\top}Y_{n,i})$ and $\zeta_{n,i}^{(k)}\eqdef c^{\top}Y_{n,i}-\xi_{n,i}^{(k)}$ so that
\begin{align*}
	\norm{c^{\top}(\bar{Y}_n-\E[\bar{Y}_n\mid \CS])}_{\CS,1}&\le 2\max_{i\in N_n}\norm{\zeta_{n,i}^{(k)}}_{\CS,1} \\
	&\quad +\norm{n^{-1}\sum_{i\in N_n}\left(\xi_{n,i}^{(k)}-\E[\xi_{n,i}^{(k)}\mid \CS]\right)}_{\CS,2} \qtext{a.s.}
\end{align*}
The result then follows from the definition of the essential infimum and the following inequalities:
\[
	\norm{\zeta_{n,i}^{(k)}}_{\CS,1}\le \E[\norm{Y_{n,i}}\ind\{\norm{Y_{n,i}}>k\}\mid \CS] \qtext{a.s.}
\]
and, since $\psi_{1,1}(\varphi_k,\varphi_k)\le Ck^2$,
\[
	\norm{\sum_{i\in N_n}\left(\xi_{n,i}^{(k)}-\E[\xi_{n,i}^{(k)}\mid \CS]\right)}_{\CS,2}\le \sqrt{n}k\left(4+C\sum_{s\ge 1}\delta_n^{\partial}(s;1)\theta_{n,s}\right)^{1/2} \qtext{a.s.} \qedhere
\]
\end{proof}

\medskip

\begin{proof}[\textbf{Proof of Theorem \ref{thm:bootstrap_consistency1}}]
The first assertion follows trivially from the triangle inequality. Consider the second assertion. First, note that for a sub-$\sigma$-algebra $\F\subset \H$ and random variables $X,Y$ and $\F$-measurable random variable $Z$, $\PR{X\le Z\mid \F}=\FD{X}{\F}(\csdot,Z)$ a.s.\ and $\PR{Y\le Z\mid \F}=\FD{Y}{\F}(\csdot,Z)$ a.s.\ \citep[see, e.g.,][Theorem~5.4]{Kallenberg:02:FMP}. Therefore,
\[
	\abs{\PR{X\le Z\mid \F}-\PR{Y\le Z\mid \F}}\le \DK{X,Y\mid \F} \qtext{a.s.}
\]
In addition, if $\F=\sigma(\A\cup \B)$, where $\A$ and $\mathcal{B}$ are sub-$\sigma$-algebras of $\H$, and $Y$ is conditionally independent of $\A$ given $\B$, then $\DK{X,Y\mid \F}=\DK{X,Y\mid \F,\B}$ a.s.

Let $c_n(\alpha)$ denote the conditional $\alpha$-quantile of $S_n$ given $\CS$. Fix $\eta>0$ such that $\alpha\pm 2\eta\in (0,1)$ and let $\Delta_n\equiv \DK{T_n^{*},S_n\mid \G_n,\CS}$. Then, using the properties of generalized inverses \citep[see, e.g.,][Proposition 1]{Embrechts/Hofert:13} and the conditional independence of $S_n$ and $Y_n$ given $\CS$, we get
\begin{align*}
	\PR{S_n\le c_n^{*}(\alpha+\eta)\mid \G_n} &\ge \PR{T_n^{*}\le c_n^{*}(\alpha+\eta)\mid \G_n}-\eta \\
	&\ge \alpha \\
	&= \PR{S_n\le c_n(\alpha)\mid \G_n}\qtext{a.s.\ on } \{\Delta_n\le \eta\}
\end{align*}
and
\begin{align*}
	\PR{T_n^{*}\le c_n(\alpha+2\eta)\mid \G_n} &\ge \PR{S_n\le c_n(\alpha+2\eta)\mid \G_n}-\eta \\
	&= \alpha+\eta \\
	&\ge \PR{T_n^{*}\le c_n^{*}(\alpha)\mid \G_n} \qtext{a.s.\ on } \{\Delta_n\le \eta\}.
\end{align*}
Therefore,
\begin{align*}
	\PR{c_n^{*}(\alpha)\ge c_n(\alpha-\eta) \mid\CS}&\ge \PR{c_n^{*}(\alpha)\ge c_n(\alpha-\eta),\Delta_n\le \eta\mid\CS} \\
	&=\PR{\Delta_n\le \eta\mid \CS} \qtext{a.s.}
\end{align*}
and
\begin{align*}
	\PR{c_n^{*}(\alpha)\le c_n(\alpha+2\eta)\mid\CS}&\ge \PR{c_n^{*}(\alpha)\le c_n(\alpha+2\eta),\Delta_n\le \eta\mid\CS} \\
	&=\PR{\Delta_n\le \eta\mid \CS} \qtext{a.s.}
\end{align*}
Using the last two inequalities we find that
\begin{align*}
	&\PR{c_n(\alpha)\wedge c_n^{*}(\alpha)<T_n\le c_n(\alpha)\vee c_n^{*}(\alpha)\mid \CS} \\
	&\qquad \le \PR{c_n(\alpha-\eta)<T_n\le c_n(\alpha+2\eta)\mid\CS} \\
	&\quadt +\PR{c_n(\alpha-\eta)>c_n^{*}(\alpha)\mid\CS}+\PR{c_n(\alpha+2\eta)<c_n^{*}(\alpha)\mid\CS} \\
	&\qquad\le \PR{c_n(\alpha-\eta)<S_n\le c_n(\alpha+2\eta)\mid\CS} \\
	&\quadt +2\PR{\Delta_n>\eta\mid\CS}+2\DK{T_n,S_n\mid \CS} \\
	&\qquad= 3\eta+2\PR{\Delta_n>\eta\mid\CS}+2\DK{T_n,S_n\mid \CS} \qtext{a.s.}
\end{align*}
and
\begin{align}
\label{eq:quantile_bound}
\begin{aligned}
	A_{n,\alpha}&\eqdef\abs{\PR{T_n\le c_n^{*}(\alpha)\mid \CS}-\alpha} \\
	&\phantom{:}\le 3\eta+2\PR{\Delta_n>\eta\mid\CS}+3\DK{T_n,S_n\mid \CS} \qtext{a.s.}
\end{aligned}
\end{align}
Finally, there exists a sequence $\seq{\alpha_k}$ such that $\esssup_{\alpha\in (0,1)} A_{n,\alpha}=\sup_k A_{n,\alpha_k}$ a.s.\ and the latter is a.s.\ bounded by the RHS of \eqref{eq:quantile_bound}. Therefore,
\[
	\limsup_{n\to\infty}\Big(\esssup_{\alpha\in (0,1)} A_{n,\alpha}\Big)\le 3\eta \qtext{a.s.}
\]
and the result follows by considering a sequence $\eta_m\searrow 0$.
\end{proof}

\medskip

\begin{proof}[\textbf{Proof of Lemma \ref{lemma:bootstrap_consistency2}}]
By the mean value theorem, we may write
\begin{align}
\label{eq:mvapprox}
\begin{aligned}
	T_n&=\nabla\phi(\tilde{\theta}_n)^{\top}\tau_n(\hat{\theta}_n-\theta_n) \qtext{and} \\
	T_n^{*}&=\nabla\phi(\tilde{\theta}_n^{*})^{\top}\tau_n^{*}(\hat{\theta}_n^{*}-\theta_n^{*}),
\end{aligned}
\end{align}
where $\tilde{\theta}_n$ and $\tilde{\theta}_n^{*}$ are such that $\normin{\tilde{\theta}_n-\theta_n}\le \normin{\hat{\theta}_n-\theta_n}$ and $\normin{\tilde{\theta}_n^{*}-\theta_n^{*}}\le \normin{\hat{\theta}_n^{*}-\theta_n^{*}}$. Then for any $r\in \R$ and $\epsilon>0$,
\begin{align*}
	&\abs{\PR{T_n^{*}\le r\mid \G_n}-\PR{T_n'^{*}\le r\mid \G_n}} \\
	&\qquad \le \PR{T_n'^{*}\le r+R_n^{*}\mid \G_n}-\PR{T_n'^{*}\le r-R_n^{*}\mid \G_n} \\
	&\qquad \le 2\DK{T_n'^{*},S_n'\mid \G_n,\CS}+\LC{S_n',2\epsilon\mid \CS}+\PR{R_n^{*}>\epsilon\mid \G_n} \qtext{a.s.},
\end{align*}
where
\[
	R_n^{*}\equiv \abs{(\nabla\phi(\tilde{\theta}_n^{*})-\nabla\phi(\theta_n^{*}))^{\top}\tau_n^{*}(\hat{\theta}_n^{*}-\theta_n^{*})}.
\]
Similarly, for any $r\in \R$ and $\epsilon>0$,
\begin{align*}
	&\abs{\PR{T_n\le r\mid \CS}-\PR{T_n'\le r\mid \CS}} \\
	&\qquad \le 2\DK{T_n',S_n'\mid \CS}+\LC{S_n',2\epsilon\mid \CS}+\PR{R_n>\epsilon\mid \CS} \qtext{a.s.},
\end{align*}
where
\[
	R_n\equiv \abs{(\nabla\phi(\tilde{\theta}_n)-\nabla\phi(\theta_n))^{\top}\tau_n(\hat{\theta}_n-\theta_n)}.
\]

By Lemma \ref{lemma:aux_asy_tightness} the sequence $\seq{\theta_n^{*}}$ is $\CS$-asymptotically tight. Therefore, using Lemma \ref{lemma:aux_CMT} together with the $\CS$-asymptotic tightness of $\tau_n^{*}(\hat{\theta}_n^{*}-\theta_n^{*})$ and $\tau_n(\hat{\theta}_n-\theta_n)$ it follows that $\PR{R_n^{*}>\epsilon\mid \CS}\to 0$ a.s.\ and $\PR{R_n>\epsilon\mid \CS}\to 0$ a.s. Consequently, for any $\nu>0$,
\begin{align*}
	&\limsup_{n\to\infty}\PR{\DK{T_n^{*},T_n'^{*}\mid \G_n}>\nu\mid \CS} \\
	&\qquad \le 3\nu^{-1}\essinf_{\epsilon>0}\limsup_{n\to\infty}\LC{S_n',\epsilon\mid \CS}=0 \qtext{a.s.}
\shortintertext{and}
	&\limsup_{n\to\infty}\DK{T_n,T_n'\mid \CS} \\
	&\qquad \le \essinf_{\epsilon>0}\limsup_{n\to\infty}\LC{S_n',\epsilon\mid \CS}=0 \qtext{a.s.}
\end{align*}
The result then follows from the triangle inequality.
\end{proof}

\medskip

\begin{proof}[\textbf{Proof of Theorem \ref{thm:bootstrap_consistency2}}]
Follows immediately from Lemma \ref{lemma:bootstrap_consistency2} and Theorem \ref{thm:bootstrap_consistency1}.
\end{proof}

\medskip

\begin{proof}[\textbf{Proof of Corollary \ref{corr:bootstrap_consistency2}}]
Consider Equation \eqref{eq:mvapprox} in the proof of Lemma \ref{lemma:bootstrap_consistency2}. By Lemma \ref{lemma:aux_weak_conv} $T_n$ converges $\CS$-weakly to $S'$ ($\because \tilde{\theta}_n\PRC{\CS}\theta$ a.s.\ and $x\mapsto\nabla\phi(x)$ is continuous). Hence, $\DK{T_n,S'\mid \CS}\to 0$ a.s.\ by Lemma \ref{lemma:aux_weak_eq}. Convergence of $\DK{T_n^{*},S'\mid \G_n,\CS}$ follows from arguments similar to those given in the proof of Lemma \ref{lemma:bootstrap_consistency2}. Finally, the result holds by Theorem \ref{thm:bootstrap_consistency1}.
\end{proof}

\medskip

\begin{proof}[\textbf{Proof of Proposition \ref{prop:BB_var}}]
Let $\zeta_{n,i}\eqdef \sum_{j\in B_{n,i}}Y_{n,i}'$, where $Y_{n,i}'\eqdef Y_{n,i}-\mu_{n,i}$, and let $\zeta_{n,i}^{*}$ be its resampling version. Then, using the conditional independence of the elements of $\{\zeta_{n,i}^{*}\}$ given $\G_n$, we find that
\begin{align}
	&\begin{aligned}
	\label{eq:BB_var1}
		\tilde{\Sigma}_n&\eqdef\Var\left(\frac{1}{\sqrt{n}}\sum_{k=1}^{K_n}\zeta_{n,k}^{*}\mid \G_n\right)=\frac{1}{n}\sum_{k=1}^{K_n}\Var(\zeta_{n,k}^{*}\mid \G_n) \\
		&=\frac{1}{\delta_n(s_n)}\left(\frac{1}{n}\sum_{i\in N_n}\zeta_{n,i}\zeta_{n,i}^{\top}-\bar{\zeta}_n\bar{\zeta}_n^{\top} \right) \\
		&=\Sigma_n^{*}-\frac{1}{n}\sum_{i\in N_n}(\omega_n(i)-1)(\zeta_{n,i}\mu_n^{\top}+\mu_n\zeta_{n,i}^{\top})-\frac{\Delta_n(s_n;2)}{\delta_n(s_n)}\mu_n\mu_n^{\top} \\
		&\equiv \Sigma_n^{*}-A_{n,1}+A_{n,2},
	\end{aligned}
\intertext{where $\bar{\zeta}_n\eqdef n^{-1}\sum_{i\in N_n}\zeta_{n,i}$. On the other hand, using the second line of \eqref{eq:BB_var1},}
	&\begin{aligned}
	\label{eq:BB_var2}
		\tilde{\Sigma}_n&=\frac{1}{n}\sum_{i,j\in N_n}\omega_n(i,j)Y_{n,i}'Y_{n,j}'^{\top} \\
		&\qquad-\delta_n(s_n)\times\frac{1}{n}\sum_{i\in N_n}\omega_n(i)Y_{n,i}'\times\frac{1}{n}\sum_{i\in N_n}\omega_n(i)Y_{n,i}'^{\top} \\
		&\equiv B_{n,1}+B_{n,2}.
	\end{aligned}
\end{align}

Let $y_{n,i}'\eqdef c^{\top}Y_{n,i}'$ and $\mu_n'\eqdef c^{\top}\mu_n$, where $c\in\R^v$ with $\norm{c}=1$ and note that by Lemma \ref{lemma:aux_matrix_conv} it suffices to show that $\E[\absin{c^{\top}(\Sigma_n^{*}-\Sigma_n)c}\mid \CS]\to 0$ a.s. Also by Lemma 2.1. in \citealias{Kojevnikov/Marmer/Song:20} the process $\seq{y_n'}$ is $(\L_1,\psi,\CS)$-weakly dependent with the weak dependence coefficients $\seq{\gamma_n}$. In the following let
\[
	\Xi_n\eqdef \sum_{s\ge 1}\delta_n^{\partial}(s)\gamma_{n,s}^{1-\frac{2}{r}}.
\]

\begin{claim}
$\E[\absin{c^{\top}(\tilde{\Sigma}_n-\Sigma_n)c}\mid \CS]\to 0$ a.s.
\end{claim}
\begin{subproof}
Consider the first term on the last line of \eqref{eq:BB_var2}. Write
\begin{align*}
	c^{\top}(B_{n,1}-\Sigma_n)c&=\frac{1}{n}\sum_{i\in N_n}(y_{n,i}'^{2}-\E[y_{n,i}'^{2}\mid \CS])+\frac{1}{n}\sum_{i\in N_n}(\omega_n(i)-1)y_{n,i}'^{2} \\
	&\quad +\frac{1}{n}\sum_{i\in N_n}\sum_{j\in N_n\setminus\{i\}}\omega_n(i,j)(y_{n,i}'y_{n,j}'-\E[y_{n,i}'y_{n,j}'\mid \CS]) \\
	&\quad +\frac{1}{n}\sum_{i\in N_n}\sum_{j\in N_n\setminus\{i\}}(\omega_n(i,j)-1)\E[y_{n,i}'y_{n,j}'\mid \CS] \\
	&\equiv R_{n,0}+R_{n,1}+R_{n,2}+R_{n,3}.
\end{align*}

Using the covariance inequalities established in \citealias{Kojevnikov/Marmer/Song:20},
\begin{equation}
\label{eq:BB_var3}
	\abs{R_{n,3}}\le C_3\sum_{s\ge 1}\gamma_{n,s}^{1-\frac{2}{r}} \times \frac{1}{n}\sum_{i\in N_n}\sum_{j\in N_n^{\partial}(i;s)}\abs{\omega_n(i,j)-1} \qtext{a.s.},
\end{equation}
where $C_3=C(\mu_{2r}^2\vee 1)$ for some constant $C\ge 1$. Since $\tilde{\omega}<\infty$ a.s., the RHS of \eqref{eq:BB_var3} is bounded by $C_3(\tilde{\omega}+1)\Xi_n<\infty$ a.s. Therefore, by the dominated convergence theorem $\abs{R_{n,3}}\to 0$ a.s. Also letting $w_{n,i,j}\eqdef y_{n,i}'y_{n,j}'-\E[y_{n,i}'y_{n,j}'\mid \CS]$, we find that
\begin{align*}
	\E[R_{n,2}^2\mid \CS]&\le \frac{\bar{\omega}^2}{n^2}\sum_{\substack{i,j\in N_n \\ 1\le d_n(i,j)< 2(s_n+1)}}\sum_{\substack{k,l\in N_n \\ 1\le d_n(k,l)< 2(s_n+1)}}\abs{\E[w_{n,i,j}w_{n,k,l}\mid \CS]} \\
	&\le \frac{C_2\bar{\omega}^2}{n^2}\sum_{s\ge 0}\abs{H_n(s,2s_n+1)}\gamma_{n,s}^{1-\frac{2}{r}}\to 0 \qtext{a.s.},
\end{align*}
where $C_2= C(\mu_{2r}^4\vee 1)$ for some constant $C\ge 1$. Finally,
\begin{align*}
	\E[\abs{R_{n,1}}\mid \CS]&\le \frac{\mu_{2r}^2}{n}\sum_{i\in N_n}\abs{\omega_n(i)-1}\to 0 \qtext{a.s.}
\shortintertext{and}
	\E[R_{n,0}^2\mid \CS]&\le \frac{1}{n^2}\sum_{i,j\in N_n}\abs{\Cov(y_{n,i}'^2,y_{n,j}'^2\mid \CS)} \\
	&\le \frac{C_0}{n}(1+\Xi_n)\to 0 \qtext{a.s.},
\end{align*}
where $C_0=C(\mu_{2r}^4\vee 1)$ for some constant $C\ge 1$.

As for the second term, note that $c^{\top}B_{n,2}c\ge 0$ a.s.\ and
\begin{align*}
	\E[c^{\top}B_{n,2}c\mid \CS]&\le \frac{(D_n(s_n))^2}{\delta_n(s_n)n^2}\sum_{i,j\in N_n}\abs{\E[y_{n,i}'y_{n,j}'\mid \CS]} \\
	&\le \frac{(D_n(s_n))^2C_3}{\delta_n(s_n)n}(1+\Xi_n)\to 0 \qtext{a.s.} \qedhere
\end{align*}
\end{subproof}

\medskip

Consider the last two terms on the last line of equation \eqref{eq:BB_var1}. Let
\[
	\alpha_{n,i}\eqdef \sum_{j\in B_{n,i}}(\omega_n(j)-1)\qtext{and}\quad \alpha_n\eqdef \max_{i\in N_n}\abs{\alpha_{n,i}}.
\]
Then, since
\[
	\sum_{i\in N_n}(\omega_n(i)-1)\zeta_{n,i}=\sum_{i\in N_n}\alpha_{n,i}Y_{n,i}',
\]
we have $c^{\top}A_{n,1}c=2n^{-1}\sum_{i\in N_n}\alpha_{n,i}y_{n,i}'\mu_n'$ and, therefore,
\[
	\E\left[(c^{\top}A_{n,1}c)^2\mid\CS\right]\le \frac{4C_2\alpha_n^2}{n}(1+\Xi_n)\to 0 \qtext{a.s.}
\]
Finally,
\[
	\norm{A_{n,2}}_{F}\le \mu_{2r}^2\times \frac{\Delta_n(s_n;2)}{\delta_n(s_n)}\to 0 \quad{a.s.} \qedhere
\]
\end{proof}

\medskip

\begin{proof}[\textbf{Proof of Proposition \ref{prop:DWB_var}}]
We use the notation from the proof of Proposition \ref{prop:local_lipshitz}. For a vector $c\in \R^{v}$ with $\norm{c}=1$ we have
\begin{align}
\label{eq:DWB_var1}
\begin{aligned}
c^{\top}(\Sigma_n^{*}-B_{n,1})c &=\bar{y}_n'^2\times\frac{\delta_n(s_n;2)}{\delta_n(s_n)}-\frac{2\bar{y}_n'}{n}\sum_{i\in N_n}y_{n,i}'\sum_{j\in N_n}\omega_n(i,j) \\
&\equiv Q_{n,1}+Q_{n,2},
\end{aligned}
\end{align}
where $\bar{y}_n\eqdef n^{-1}\sum_{i\in N_n}y_{n,i}'$. Consider the second term in the last line of \eqref{eq:DWB_var1}. Letting $\tau_{n,i}\eqdef \sum_{j\in N_n}\omega_n(i,j)$ and noticing that
\begin{align*}
	\max_{i\in N_n}\tau_{n,i}&\le\max_{i\in N_n}\left(\omega_n(i)+\tilde{\omega}\abs{N_n(i;s_n+1)}\right) \\
	&\le (\tilde{\omega}+1)D_n(s_n)\equiv \tau_n,
\end{align*}
it follows that
\[
	\abs{Q_{n,2}}\le 2\abs{\sqrt{\tau_n}\bar{y}_n'}\times \frac{1}{n\sqrt{\tau_n}}\abs{\sum_{i\in N_n}\tau_{n,i}y_{n,i}'}\equiv 2\abs{\sqrt{\tau_n}\bar{y}_n'}\times Q_{n,3}.
\]
Similarly to the proof of Proposition \ref{prop:BB_var},
\[
	\E[Q_{n,3}^2\mid\CS]\le \frac{C_3\tau_n}{n}(1+\Xi_n)\to 0 \qtext{a.s.}
\]
Finally, $\E[\tau_n\bar{y}_n'^2\mid \CS]$ is bounded the same quantity and since $\delta_n(s_n;2)/\delta_n(s_n)\le \tau_n$,
\[
	\E[Q_{n,1}\mid \CS]\le \E[\tau_n\bar{y}_n'^2\mid \CS]\to 0 \qtext{a.s.} \qedhere
\]
\end{proof}

\medskip

\begin{proof}[\textbf{Proof of Proposition \ref{prop:BB_consistency}}]
Consider $T_{1,n}^{*}$ first. Let $\tilde{Y}_{n,i}\eqdef Y_{n,i}-\mu_n$, $\tilde{Z}_{n,i}\eqdef \sum_{k\in B_{n,i}}\tilde{Y}_{n,k}$ and let $\tilde{Z}_{n,i}^{*}$ be the bootstrap version of the letter. Also define $W_{n,i}^{*}\eqdef \tilde{Z}_{n,i}^{*}-\E[\tilde{Z}_{n,i}^{*}\mid \G_n]$. Conditionally on $\G_n$, $\arr{\xi_{n,i}^{*}}$ are row-wise i.i.d.\ random vectors with $\E[W_{n,1}^{*}\mid \G_n]=0$ and $\Var(W_{n,1}^{*}\mid \G_n)=\delta_n(s_n)\Sigma_n^{*}$ a.s. Write
\[
	\sqrt{n}(\tilde{Y}_n^{*}-\mu_n^{*})=\frac{1}{\sqrt{K_n}}\sum_{k=1}^{K_n}(\delta_n(s_n))^{-1/2}W_{n,k}^{*}.
\]
Then, letting $\underline{\lambda}(A)$ denote the minimal eigenvalue of a square matrix $A$, by Corollary \ref{corr:aux_CLT1},
\[
	\DK{T_{1,n}^{*},S_{1,n}^{*}\mid \G_n}\le \frac{C_v}{(\underline{\lambda}(\Sigma)/2)^{3/8}}\left(\frac{\E[\normin{W_{n,1}^{*}}^3\mid \G_n]}{\sqrt{n}\delta_n(s_n)}\right)^{1/4}
\]
a.s.\ on $\{\underline{\lambda}(\Sigma_n^{*})\ge \underline{\lambda}(\Sigma)/2\}$, where $S_{1,n}^{*}=\norm{Q_n}$ and $Q_n$ is conditionally normal given $\G_n$ with zero mean and variance $\Sigma_n^{*}$.

\begin{claim}
$\E[\normin{W_{n,i}^{*}}^3\mid \CS]/(\sqrt{n}\delta_n(s_n))\to 0$ a.s.
\end{claim}

\begin{subproof}
It suffices to show that $\E[\absin{c^{\top}W_{n,i}}^3\mid \CS]/(\sqrt{n}\delta_n(s_n))\to 0$ a.s.\ for any $c\in \R^v$ such that $\norm{c}=1$. By the $c_r$-inequality,\footnote{
	We consider the conditional versions of all inequalities used in this proof.
}
\[
	\E[\absin{c^{\top}W_{n,1}^{*}}^3\mid \G_n]\le 8\E[\absin{c^{\top}\tilde{Z}_{n,1}^{*}}^3\mid \G_n]=\frac{8}{n}\sum_{i\in N_n}\absin{c^{\top}\tilde{Z}_{n,i}}^3 \qtext{a.s.}
\]
Let $\tilde{y}_{n,i}\eqdef c^{\top}\tilde{Y}_{n,i}$. Then
\begin{align*}
	\E[\absin{c^{\top}\tilde{Z}_{n,i}}^4\mid \CS]&\le \sum_{j_1,j_2,j_3,j_4\in B_{n,i}}\abs{\Cov(\tilde{y}_{n,j_1}\tilde{y}_{n,j_2},\tilde{y}_{n,j_3}\tilde{y}_{n,j_4}\mid \CS)} \\
	&+\left(\sum_{j_1,j_2\in B_{n,i}}\abs{\Cov(\tilde{y}_{n,j_1},\tilde{y}_{n,j_2}\mid \CS)}\right)^2\equiv A_{n,i}+C_{n,i}^2 \qtext{a.s.}
\end{align*}
Similarly to the proof of Proposition \ref{prop:BB_var}, we find that w.p.1,
\begin{align*}
	A_{n,i}&\le C_1(\tilde{\mu}_{2p}^4\vee 1)\abs{B_{n,i}}^3\sum_{s\ge 0}h_{loc,n}(s,s_n)\gamma_{n,s}^{1-\frac{2}{p}} \qtext{and} \\
	C_{n,i}&\le C_2(\tilde{\mu}_{2p}^2\vee 1)\abs{B_{n,i}}\sum_{s\ge 0}\delta_{loc,n}^{\partial}(s,s_n)\gamma_{n,s}^{1-\frac{2}{p}},
\end{align*}
where $C_1$ and $C_2$ are some positive constants. The result then follows by noticing that $\E[\absin{c^{\top}\tilde{Z}_{n,i}}^3\mid \CS]\le (\E[\absin{c^{\top}\tilde{Z}_{n,i}}^4\mid \CS])^{3/4}$ a.s.\ and the fact that $(A_{n,i}+C_{n,i}^2)^{3/4}\le A_{n,i}^{3/4}+C_{n,i}^{3/2}$.
\end{subproof}

\medskip

Using Jensen's inequality, we find that for any $\epsilon>0$,
\begin{align*}
	&\PR{\DK{T_{1,n}^{*},S_{1,n}^{*}\mid \G_n}>\epsilon\mid \CS} \\
	&\qquad\le \frac{C_v}{(\underline{\lambda}(\Sigma)/2)^{3/8}}\left(\frac{\E[\normin{W_{n,1}^{*}}^3\mid \CS]}{\sqrt{n}\delta_n(s_n)}\right)^{1/4}+\PR{\underline{\lambda}(\Sigma_n^{*})< \underline{\lambda}(\Sigma)/2\mid \CS} \\
	&\qquad \to 0 \qtext{a.s.},
\end{align*}
where the convergence of $\PR{\underline{\lambda}(\Sigma_n^{*})< \underline{\lambda}(\Sigma)/2\mid \CS}$ follows from the fact that the eigenvalues of a matrix depend continuously on the entries of the matrix \citep[see, e.g.,][Theorem~2.11]{Zhang:11:MatrixTheory} so that $\lambda_j(\Sigma_n^{*})\PRC{\CS}\lambda_j(\Sigma)$ a.s.\ for all $1\le j\le d$.

Since the eigenvalues of $\Sigma_n^{*}$ converge to the eigenvalues of $\Sigma$ and the latter are a.s.\ positive, it follows from Lemma \ref{lemma:aux_normal_approx1} that
\[
	\DK{S_{1,n}^{*},\normin{\Sigma^{1/2}\eta}\mid \G_n,\CS}\PRC{\CS}0 \qtext{a.s.}
\]
Finally, by Lemma \ref{lemma:aux_weak_conv}, $T_{1, n}\to \normin{\Sigma^{1/2}\eta}$ \, $\CS$-weakly and, hence, the result follows from Corollary \ref{corr:bootstrap_consistency1}.

\medskip

Consider the second assertion. First, for any $c\in \R^v$ such that $\norm{c}=1$ and $\epsilon>0$, we get
\begin{align*}
	\PR{\absin{c^{\top}(\tilde{Y}_n-\mu_n^{*})}>\epsilon\mid \CS}&\le \frac{1}{(K_n\epsilon)^2}\E\left[\left(\sum_{k=1}^{K_n} c^{\top}\tilde{Z}_{n,k}^{*}\right)^2\mid \CS\right] \\
	&=\frac{1}{K_n\epsilon^2}\E[c^{\top}\Sigma_n^{*}c\mid \CS]\to 0 \qtext{a.s.},
\end{align*}
where the convergence follows from the consistency of $\Sigma_n^{*}$. Therefore, $\tilde{Y}_n^{*}-\mu_n^{*}\PRC{\CS}0$ a.s. Also since the $\CS$-asymptotic tightness of a vector follows from that of its elements, $\sqrt{n}(\tilde{Y}_n^{*}-\mu_n^{*})$ is $\CS$-asymptotically tight due to the same reason (i.e., the convergence of $\Sigma_n^{*}$). Write
\[
	\nabla\phi(\mu_n^{*})^{\top}\sqrt{n}(\tilde{Y}_n^*-\mu_n^{*})=\frac{1}{\sqrt{K_n}}\sum_{k=1}^{K_n}(\delta_n(s_n))^{-1/2}\nabla\phi(\mu_n^{*})^{\top}W_{n,k}^{*}.
\]
Then by Lemma \ref{lemma:aux_CLT2}, letting $T_{2,n}'^{*}=\nabla\phi(\mu_n^{*})^{\top}\sqrt{n}(\tilde{Y}_n^{*}-\mu_n^{*})$ and $S_{2,n}'^{*}=\nabla\phi(\mu_n^{*})^{\top}Q_n$, we have
\[
	\DK{T_{2,n}'^{*},S_{2,n}'^{*}\mid \G_n}\le \frac{C}{(\sigma/2)^3}\times\frac{\norm{\nabla\phi(\mu_n^{*})}^3\E[\normin{W_{n,1}^{*}}^3\mid \G_n]}{\sqrt{n}\delta_n(s_n)}
\]
a.s.\ on $\{\sigma_n^{*}\ge \sigma/2\}$, where $\sigma_n^{*2}=\nabla\phi(\mu_n^{*})^{\top}\Sigma_n^{*}\nabla\phi(\mu_n^{*})$ and $\sigma^{2}=\nabla\phi(\mu)^{\top}\Sigma\nabla\phi(\mu)$. Since $x\mapsto\nabla\phi(x)$ is continuous and $\mu_n^{*}$ is a consistent estimator of $\mu$, $\nabla\phi(\mu_n^{*})\PRC{\CS}\nabla\phi(\mu)$ and $\sigma_n^{*}\PRC{\CS}\sigma$ a.s. Consequently, as in the previous case,
\[
	\DK{T_{2,n}'^{*},S_{2,n}'^{*}\mid \G_n}\PRC{\CS}0 \qtext{a.s.}
\]
and by Lemma \ref{lemma:aux_normal_approx2},
\[
	\DK{S_{2,n}'^{*},\nabla\phi(\mu)^{\top}\Sigma^{1/2}\eta\mid \G_n,\CS}\PRC{\CS}0 \qtext{a.s.}
\]
The result then follows from Corollary \ref{corr:bootstrap_consistency2}.
\end{proof}

\medskip

\begin{proof}[\textbf{Proof of Proposition \ref{prop:DWB_consistency1}}]
The proof is similar to one for Proposition \ref{prop:BB_consistency}, and so is omitted.
\end{proof}

\medskip

\begin{proof}[\textbf{Proof of Proposition \ref{prop:DWB_consistency2}}]
By Lemma \ref{lemma:aux_normal_approx3},
\begin{align*}
	&\DK{\norm{T_{1,n}^{*}},\norm{S_{1,n}^{*}}\mid \G_n}\PRC{\CS}0 \qtext{a.s.\ and} \\
	&\DK{T_{2,n}'^{*},S_{2,n}'^{*}\mid \G_n}\PRC{\CS}0 \qtext{a.s.},
\end{align*}
where $S_{1,n}^{*}=\norm{Q_n}$, $S_{2,n}'^{*}=\nabla\phi(\bar{Y}_n)^{\top}Q_n$ and $Q_n$ is conditionally normal given $\G_n$ with zero mean and variance $\Sigma_n^{*}$. The rest is similar to the proof of Proposition \ref{prop:BB_consistency}.
\end{proof}

\section{Network HAC Estimator}
\label{section:app_HAC}

Although the HAC estimator \eqref{eq:HAC} is consistent in the sense that $\hat{\Sigma}_n-\Sigma_n\PRC{\CS}0$ a.s., it does not necessarily yields positive semi-definite covariance matrix. There exist a number methods of approximating a symmetric matrix by a positive definite matrix \citep[see, e.g.,][]{Higham:88,Higham:02}. Borrowing some ideas from that literature we suggest a simple way of obtaining a positive definite estimate.

Let $Q_n\Lambda_nQ_n^{\top}$ be the eigendecomposition of $\hat{\Sigma}_n$ (since $\hat{\Sigma}_n$ is symmetric all its eigenvalues are real). Also let $\underline{\lambda}(A)$ denote the smallest eigenvalue of $A$, e.g., $\underline{\lambda}(\hat{\Sigma}_n)=\min_{1\le k\le v}\Lambda_n$. Consider a sequence of small positive real numbers $c_n\searrow 0$. We approximate $\hat{\Sigma}_n$ by
\[
	\hat{V}_n^{+}\eqdef Q_n\left(\Lambda_n \vee c_nI_v\right)Q_n^{\top},
\]
where the maximum is taken element-wise. By construction, the matrix $\hat{\Sigma}_n^{+}$ is positive definite. Moreover, in the case when the smallest eigenvalue of $\Sigma_n$ is bounded from below by some positive constant, it is also a consistent estimator of the true variance as follows from the next result.

\begin{prop}
\label{prop:aux_hac1}
Suppose that $\hat{\Sigma}_n-\Sigma_n\PRC{\CS}0$ a.s.\ and there exists a constant $c>0$ such that $\PR{\underline{\lambda}(\Sigma_n)\ge c \text{ ev.}}=1$. Then
\[
	\hat{\Sigma}_n^{+}-\Sigma_n\PRC{\CS}0\qtext{a.s.}
\]
\end{prop}

\begin{proof}
Fix $\epsilon>0$. Then
\begin{align}
\label{eq:approx_pr_bound}
	\begin{aligned}
		\PR{\normin{\hat{\Sigma}_n^{+}-\Sigma_n}>\epsilon\mid \CS}&\le \PR{\normin{\hat{\Sigma}_n-\Sigma_n}>\epsilon\mid \CS} \\
		&\qquad +\PR{\underline{\lambda}(\hat{\Sigma}_n)<c_n\mid \CS} \qtext{a.s.}
	\end{aligned}
\end{align}
The first term on the RHS of \eqref{eq:approx_pr_bound} trivially converges to $0$ a.s. As for the second term, using the properties of the Rayleigh quotient,
\[
	\underline{\lambda}(\hat{\Sigma}_n)=\min_{x:\norm{x}=1}x^{\top}\hat{\Sigma}_n x\ge \underline{\lambda}(\Sigma_n)+\underline{\lambda}(\hat{\Sigma}_n-\Sigma_n).
\]
Therefore, noticing that $\abs{\underline{\lambda}(A)}\le \normin{A}$,
\begin{align*}
	\PR{\underline{\lambda}(\hat{\Sigma}_n)<c_n\mid \CS}&\le \PR{\underline{\lambda}(\hat{\Sigma}_n-\Sigma_n)<c_n-c\mid \CS} \\
	&\qquad +\ind\{\underline{\lambda}(\Sigma_n)< c\}\to 0\qtext{a.s.} \qedhere
\end{align*}

\end{proof}

\medskip

If $\Sigma_n$ converges a.s.\ to a positive definite matrix $\Sigma$, then we may relax the assumptions of the preceding result.

\begin{prop}
\label{prop:aux_hac2}
Suppose that $\hat{\Sigma}_n-\Sigma\PRC{\CS}0$ a.s., where $\Sigma$ is positive definite. Then
\[
	\hat{\Sigma}_n^{+}-\Sigma\PRC{\CS}0\qtext{a.s.}
\]
\end{prop}

\begin{proof}
As in the proof of Proposition \ref{prop:aux_hac1} for any $\epsilon>0$,
\[
	\limsup_{n\to\infty}\PR{\normin{\hat{\Sigma}_n^{+}-\Sigma}>\epsilon\mid \CS}\le \essinf_{c>0}\ind\{\underline{\lambda}(\Sigma)< c\}= 0\qtext{a.s.} \qedhere
\]

\end{proof}

\section{Auxiliary Results}
\label{sec:app_aux}

In the following we assume that all random elements are defined on a common probability space $(\Omega,\PM,\H)$. Also for a vector $x\in \R^v$ let $\norm{x}$ denote the Euclidean norm of $x$ and let $\norm{\csdot}_{e,p}$ be the element-wise $p$-norm in $\R^{a\times b}$, i.e., $\norm{A}_{e,p}\eqdef \norm{\vecm(A)}_p$.

\medskip

\begin{lemma}
\label{lemma:aux_matrix_conv}
Let $A_n$ be a sequence of symmetric matrices in $\R^{v\times v}$ and $\F\subset \H$. Then the following are equivalent:
\begin{itemize}[leftmargin=1.65em]
	\item[\tn{(a)}] $\E[\norm{A_n}_{e,1}\mid \CS]\to 0$ a.s.
	\item[\tn{(b)}] $\E[\norm{A_n}_F\mid \CS]\to 0$ a.s.
	\item[\tn{(c)}] $\E[\absin{c^{\top}A_n c}\mid \CS]\to 0$ a.s.\ for any $c\in \R^v$ such that $\norm{c}=1$.
\end{itemize}
\end{lemma}
\begin{proof}
(a) is equivalent to $(b)$ because
\[
	\norm{A_n}_F\le \norm{A_n}_{e,1}\le v^2\norm{A_n}_F.
\]
The equivalence of (a) and (c) follows from the next inequalities:
\[
	\absin{c^{\top}A_n c}\le \norm{c}_{\infty}^2\norm{A_n}_{e,1}.
\]
and, letting $z_{ij}^{+}=(e_i+e_j)/\sqrt{2}$ and $z_{ij}^{-}=(e_i-e_j)/\sqrt{2}$, where $\{e_1,\ldots,e_v\}$ is a standard basis for $\R^v$,
\[
	\norm{A_n}_{e,1}\le\frac{1}{2}\sum_{i,j=1}^v\left(\absin{z_{ij}^{+\top}A_nz_{ij}^{+}}+\absin{z_{ij}^{-\top}A_nz_{ij}^{-}}\right). \qedhere
\]
\end{proof}

\medskip

The following is a simple extension of Lemma A.3. in \cite{Crimaldi:09} to the multidimensional case. For a random vector $X\in \R^v$ and $\F\subset\H$ let $\QM_X^{\F}$ denote the regular conditional distribution of $X$ given $\F$ and let $\hat{\varphi}_X$ be the corresponding characteristic functions, i.e., for $t\in\R^v$,
\[
	\hat{\varphi}_X(\omega,t)=\int \exp(it^{\top}x)\QM_X^{\F}(\omega,dx).
\]
Also the conditional characteristic function of $X$ given $\F$ is given by
\[
	\varphi_X(t\mid \F)\eqdef \E[\exp(it^{\top}X)\mid \F]
\]
and for a fixed $t\in\R^v$ and almost all $\omega\in\Omega$, $\hat{\varphi}_X(\omega,t)=\varphi_X(t\mid\F)(\omega)$.

\medskip

\begin{lemma}
\label{lemma:aux_ch_fun}
Let $\seq{X_n}$ be a sequence of random vectors in $\R^v$ and $\F\subset \H$. Then $X_n\to X$\, $\F$-weakly, i.e., for almost all $\omega\in\Omega$, $\QM_{X_n}^{\F}(\omega,\csdot)\to \QM_X^{\F}(\omega,\csdot)$ weakly, iff for every $t\in \R^v$, $\hat{\varphi}_{X_n}(\csdot,t)\to \hat{\varphi}_X(\csdot,t)$ a.s.
\end{lemma}

\medskip

The next lemma provides a number of useful properties of the almost sure conditional convergence which are typical of the usual weak convergence.

\begin{lemma}
\label{lemma:aux_weak_conv}
Let $\seq{X_n}$ and $\seq{Y_n}$ be sequences of random vectors in $\R^v$ and $\R^w$, respectively, and $\F\subset \H$. Then
\begin{itemize}[leftmargin=1.65em]
	\item[\tn{(a)}] If $X_n\to X$\, $\F$-weakly and $g:\R^v\to \R^d$ is continuous, then $g(X_n)\to g(X)$\, $\F$-weakly.
	\item[\tn{(b)}] $X_n\to X$\, $\F$-weakly iff $s^{\top}X_n\to t^{\top}X$\, $\F$-weakly for all $s\in \R^v$.
	\item[\tn{(c)}] If $Y_n\PRC{\F}Y$ a.s., where $Y$ is $\F$-measurable, then $Y_n\to Y$\, $\F$-weakly.
	\item[\tn{(d)}] Let $v=w$. If $X_n\to X$\, $\F$-weakly and $X_n-Y_n\PRC{\F}0$ a.s., then $Y_n\to X$\, $\F$-weakly.
	\item[\tn{(e)}] If $X_n\to X$\, $\F$-weakly, $Y_n\PRC{\F}Y$ a.s., where $Y$ is $\F$-measurable, then $(X_n^{\top},Y_n^{\top})\to (X^{\top},Y^{\top})$\, $\F$-weakly.
\end{itemize}
\end{lemma}

\begin{proof}
$(a)$ This follows from Lemma \ref{lemma:aux_ch_fun} and the fact that $x\mapsto \exp(it^{\top}g(x))$ is a bounded, continuous function.

\noindent (b) The sufficiency follow from part (a) because $x\mapsto s^{\top}x$ is continuous. For the necessity, suppose that all linear combinations converge $\F$-weakly. Then
\[
	\varphi_{X_n}(t\mid \F)=\varphi_{t^{\top}X_n}(1\mid \F)\to\varphi_{t^{\top}X}(1\mid \F)=\varphi_{X}(t\mid \F) \qtext{a.s.}
\]
and the result follows from Lemma \ref{lemma:aux_ch_fun}.

\noindent (c) Since for any $t\in \R^w$ and $\epsilon>0$, $\absin{e^{it^{\top}(Y_n-Y)}-1}\le \epsilon$ on $\absin{t^{\top}(Y_n-Y)}\le \epsilon$, we have
\begin{align*}
	\abs{\varphi_{Y_n}(t\mid \F)-\varphi_Y(t\mid \F)}&\le \E\left[\abs{e^{it^{\top}(Y_n-Y)}-1}\mid \F\right] \\
	&\le \epsilon+\PR{\absin{t^{\top}(Y_n-Y)}>\epsilon\mid \F} \qtext{a.s.}
\end{align*}
Therefore,
\[
	\limsup_{n\ge 1}\abs{\varphi_{Y_n}(t\mid \F)-\varphi_Y(t\mid \F)}\le \epsilon \qtext{a.s.}
\]
The result follows by considering a sequence $\epsilon_m\searrow 0$ and Lemma \ref{lemma:aux_ch_fun}.

\noindent (d) Similarly to part (c), for any $t\in \R^v$,
\[
	\abs{\varphi_{Y_n}(t\mid \F)-\varphi_X(t\mid \F)}\le \E\left[\abs{e^{it^{\top}(Y_n-X_n)}-1}\mid \F\right]\to 0 \qtext{a.s.}
\]

\noindent (e) Since $(X_n^{\top},Y^{\top})\to (X^{\top},Y^{\top})$\, $\F$-weakly, the result follows from part (d).
\end{proof}

\medskip

\begin{lemma}
\label{lemma:aux_weak_eq}
Let $\seq{X_n}$ be a sequence of random variables, $\F\subset \H$, and let $X$ be a random variable with \tn{(}a.s.\tn{)} continuous conditional cdf given $\F$ \tn{(}i.e., the map $t\mapsto \FD{X}{\F}(\omega,t)$ is continuous for \tn{(}almost\tn{)} all $\omega\in \Omega$\tn{)}. Then $X_n\to X$ $\F$-weakly iff $\DK{X_n,X\mid \F}\to 0$ a.s.
\end{lemma}
\begin{proof}
The necessity holds by Theorem 3.1.2 in \cite{Shiryaev:16:Probability} because $\DK{X_n,X\mid \F}\to 0$ a.s.\ implies that for almost all $\omega\in \Omega$ the regular conditional cdfs converge and the sufficiency follows from the $\omega$-wise application of P\'olya's theorem \citep[e.g.,][Theorem~9.1.4]{Athreya:2006:MTP}.
\end{proof}

\medskip

\begin{lemma}
\label{lemma:aux_CMT}
Suppose that $f:\R^v\to\R^w$ is continuous and $\seq{X_n}$ and $\seq{Y_n}$ are sequences of random vectors in $\R^v$ such that $Y_n-X_n\PRC{\F}0$ a.s.\ for some $\F\subset\H$ and $\seq{X_n}$ is $\F$-asymptotically tight. Then
\[
	f(Y_n)-f(X_n)\PRC{\F}0 \qtext{a.s.}
\]
\end{lemma}

\begin{proof}
For any $z>0$, the restriction $f|_{\overline{B(0,z)}}$ is uniformly continuous, i.e., $\forall\epsilon>0$, $\exists\delta_{\epsilon}>0$ such that for all $x,y\in \overline{B(0,z)}$, $\norm{f(x)-f(y)}<\epsilon$ whenever $\norm{x-y}<\delta_{\epsilon}$. Fix $\epsilon>0$. Then
\begin{align*}
	\PR{\norm{f(Y_n)-f(X_n)}>\epsilon\mid \F}&\le\PR{\norm{Y_n-X_n}>\delta_{\epsilon}\mid \F} \\
	&\quad+\PR{\norm{Y_n}>x\mid \F}+\PR{\norm{X_n}>z\mid \F} \\
	&\le 2\PR{\norm{Y_n-X_n}>\delta_{\epsilon}\wedge z/2\mid \F} \\
	&\quad +2\PR{\norm{X_n}>z/2\mid \F} \qtext{a.s.}
\end{align*}
Therefore,
\begin{align*}
	&\limsup_{n\to\infty}\PR{\norm{f(Y_n)-f(X_n)}>\epsilon\mid \F} \\
	&\qquad\le 2\essinf_{z>0}\limsup_{n\to\infty}\PR{\norm{X_n}>z\mid \F}=0 \qtext{a.s.} \qedhere
\end{align*}
\end{proof}

\medskip

\begin{lemma}
\label{lemma:aux_asy_tightness}
Suppose that $\seq{X_n}$ and $\seq{Y_n}$ are sequences of random vectors in $\R^v$ such that $X_n$ is $\F$-measurable for all $n\ge 1$ and some $\F\subset\H$, $\sup_{n}\norm{X_n}<\infty$ a.s., and $Y_n-X_n\PRC{\F}0$ a.s. Then $\seq{Y_n}$ is $\F$-asymptotically tight.
\end{lemma}

\begin{proof}
For any $y>0$,
\begin{align*}
	\PR{\norm{Y_n}>y\mid \F}&\le \PR{\norm{Y_n-X_n}>y/2\mid \F} \\
	&\quad+\ind\big\{\sup_{n}\norm{X_n}>y/2\big\} \qtext{a.s.}
\end{align*}
Therefore,
\begin{align*}
	&\essinf_{y>0}\limsup_{n\to\infty}\PR{\norm{Y_n}>y\mid \F} \\
	&\qquad\le \essinf_{y>0}\ind\big\{\sup_{n}\norm{X_n}>y\big\}=0 \qtext{a.s.} \qedhere
\end{align*}
\end{proof}

\medskip

In the following, for $r,\epsilon\ge 0$ let
\[
	S_{r,\epsilon}\eqdef \{x\in \R^v:r\le \norm{x}\le r+\epsilon\}.
\]

\medskip

\begin{lemma}
\label{lemma:chi_bound}
Suppose that $Z$ is a standard normal random vector in $\R^v$ with $v\ge 2$ and $\lambda_1\ge \lambda_2\ge \cdots\ge \lambda_v>0$ are constants. Let $\Lambda\eqdef\diag(\lambda_1,\dots,\lambda_v)$ and $N=\Lambda^{1/2}Z$. Then for all $\epsilon\ge 0,r\ge 0$,
\begin{align*}
	\LC{\epsilon,\normin{N}}&=\sup_{r\ge 0}\PR{N\in S_{r,\epsilon}^1}\le \frac{C_d\epsilon}{\sqrt{\lambda_2}},
\end{align*}
where $C_v\equiv\sqrt{v-1}$.
\end{lemma}
\begin{proof}
Let $X\eqdef\sum_{i=1}^2 N_i^2$, $Y\eqdef\sum_{i=3}^d N_i^2$, and note that $\normin{N}=\sqrt{X+Y}$. Then letting $f_X$ denote the density of $X$ we have
\begin{align}
\label{eq:aux_pdf2_bound}
\begin{aligned}
	f_X(x)&=\frac{1}{2\pi\sqrt{\lambda_1\lambda_2}}\int_0^x e^{-\frac{z}{2\lambda_2}-\frac{x-z}{2\lambda_1}}(z(x-z))^{-1/2}dz \\
	&\le\frac{1}{2\pi\sqrt{\lambda_1\lambda_2}}B(1/2,1/2)e^{-\frac{x}{2\lambda_1}}.
\end{aligned}
\end{align}
For $y\ge 0$ the density of $\sqrt{X+y}$ is zero on $(-\infty,\sqrt{y})$ and using \eqref{eq:aux_pdf2_bound} it can be bounded on $[\sqrt{y},\infty)$ by
\[
	f_{\sqrt{X+y}}(x)=2xf_X(x^2-y)\le\frac{\sqrt{y+\lambda_1}}{\sqrt{\lambda_1\lambda_2}},
\]
so that for all $r\ge 0$,
\[
	g(y)=\PR{r\le\sqrt{X+y}\le r+\epsilon}\le \frac{\sqrt{y+\lambda_1}}{\sqrt{\lambda_1\lambda_2}}\epsilon.
\]
Hence, for $d\ge 3$, noticing that $X$ is independent of $Y$, we find that
\begin{align*}
	&\PR{r\le \sqrt{X+Y}\le r+\epsilon}=\E[g(Y)] \\
	&\qquad\le \frac{\epsilon}{\sqrt{\lambda_1\lambda_2}}\E[Y+\lambda_1]^{1/2}
	\le \frac{\epsilon}{\sqrt{\lambda_2}}\left(\sum_{i=3}^d\frac{\lambda_i}{\lambda_1}+1\right)^{1/2},
\end{align*}
which proves the result.
\end{proof}

\medskip

Let $\phi(w)\eqdef \norm{w}$. This function is trice continuously differentiable on $\R^v\setminus \{0\}$ and the following bounds on the derivatives of $\phi$ hold:
\begin{align}
\label{eq:norm_drv_bounds}
\begin{aligned}
	\abs{\phi'(w)(x)}&\le \norm{x} \\
	\abs{\phi''(w)(x,y)}&\le 2\norm{w}^{-1}\norm{x}\norm{y} \\
	\abs{\phi'''(w)(x,y,z)}&\le 5\norm{w}^{-2}\norm{x}\norm{y}\norm{z}.
\end{aligned}
\end{align}
For a real symmetric matrix $B$ we denote the $j$-th order statistic of its eigenvalues by $\lambda_{(j)}(B)$. Finally, we say that a random vector $X$ is conditionally normal given $\F\subset\H$ with zero mean and the conditional covariance matrix $V$, denoted by $X\mid\F\sim \ND{0}{V}$, if $V$ is $\F$-measurable, a.s.\ \textit{finite and positive semi-definite} and the conditional characteristic function of $X$ is given by
\[
	\E[e^{it^{\top}X}\mid \F]=\exp\left(-\frac{1}{2}t^{\top}Vt\right) \qtext{a.s.}
\]

\medskip

\begin{theorem}
\label{thm:aux_CLT1}
Let $X_1,\ldots,X_n$ be random vectors in $\R^v$ that are conditionally independent given $\F\subset\H$ with $\E[X_i\mid \F]=0$ and $\E[\norm{X_i}^3\mid \F]<\infty$ a.s. Let $T\eqdef \sum_{i=1}^n X_i$ and let $N$ be a random vector in $\R^v$ such that $N\mid \F\sim \ND{0}{V}$, where $V=\E[TT^{\top}\mid \F]$ a.s. Then, assuming that $\upsilon\equiv\lambda_{(d\vee 2-1)}(V)>0$ a.s.,
\begin{equation*}
	\DK{\norm{T},\norm{N}\mid \F}\le C_d\left(\upsilon^{-3/2}\sum_{i=1}^n\E[\norm{X_i}^3\mid \F]\right)^{1/4} \qtext{a.s.},
\end{equation*}
where $C_d>0$ is a constant depending only on $d$.
\end{theorem}

\begin{proof}
Let $f$ be a trice continuously differential function, such that $f(x)=1$ if $x\le 0$, $f=0$ if $x\ge\epsilon>0$, and $\abs{f^{(j)}(x)}\le D\epsilon^{-j}\ind_{(0,\epsilon)}(x)$ for some constant $D>0$ and $1\le j\le 3$. Also define
\[
	g_r(s)\eqdef f(\norm{s}-r).
\]
First,
\begin{align*}
	&\PR{\norm{T}\le r\mid \F}\le \E[g_r(T)\mid \F] \\
	&\qquad \le\PR{\norm{N}\le r+\epsilon\mid \F}+\E[g_r(T)-g_r(N)\mid \F]
\shortintertext{and}
	&\PR{\norm{T}> r\mid \F}\le 1-\E[g_{r-\epsilon}(T)\mid \F] \\
	&\qquad \le\PR{\norm{N}> r-\epsilon\mid \F}+\E[g_{r-\epsilon}(N)-g_{r-\epsilon}(T)\mid \F]
\end{align*}
a.s.\ for all $r\ge 0$ and $\epsilon>0$. Therefore, w.p.1,
\begin{align}
\label{eq:aux_CLT1}
\begin{aligned}
	&\DK{\norm{T},\norm{N}\mid \F} \\
	&\qquad =\sup_{q\in \Q_{\ge 0}}\abs{\PR{\norm{T}\le q\mid \F}-\PR{\norm{N}\le q\mid \F}} \\
	&\qquad \le \sup_{q\in \Q_{> 0}}\abs{\E[g_q(T)-g_q(N)\mid \F]}+\sup_{q\in \Q_{\ge 0}}\PR{N\in S_{q,\epsilon}\mid \F},
\end{aligned}
\end{align}

Consider the first term on the third line of \eqref{eq:aux_CLT1}.

\begin{claim}
There exists a constant $B>0$ such that for any $q>0$,
\[
	\abs{\E[g_q(T)-g_q(N)\mid \F]}\le \frac{B}{\epsilon^3}\sum_{i=1}^n \E[\norm{X_i}^3\mid \F] \qtext{a.s.}
\]
\end{claim}
\begin{subproof}
Let $Z_1,\ldots,Z_n$ be i.i.d.\ standard normal random vectors in $\R^v$ independent of $X_1,\ldots,X_n$ and $\F$ and let $Y_i\eqdef V_i^{1/2}Z_i$, where $V_i$ is a version of $\E[X_iX_i^{\top}\mid \F]$. Define
\begin{align*}
	U_i&\eqdef \sum_{k=1}^{i-1} X_k+\sum_{k=i+1}^n Y_k \\
\shortintertext{and}
	W_i&\eqdef g_q\left(U_i+X_i\right)-g_q\left(U_i+Y_i\right).
\end{align*}
Then $g_q(T)-g_q(N)=\sum_{i=1}^n W_i$ and
\[
	\abs{\E[g_q(T)-g_q(N)\mid \F]}\le \sum_{i=1}^n\abs{\E[W_i\mid \F]} \qtext{a.s.}
\]
Let $\G_{i}\eqdef \F\bigvee \sigma(X_1,\ldots,X_{i-1},Z_{i+1},\ldots,Z_n)$ and let $h_{i1}(\lambda)\eqdef g_q\left(U_i+\lambda X_i\right)$ and $h_{i2}(\lambda)\eqdef g_q\left(U_i+\lambda Y_i\right)$. Using Taylor expansion up to the third order, we find that
\[
	W_i=\sum_{j=0}^2 \frac{1}{j!}\left(h_{i1}^{(j)}(0)-h_{i2}^{(j)}(0)\right)+\frac{1}{3!}\left(h_{i1}^{(3)}(\lambda_1)-h_{i2}^{(3)}(\lambda_2)\right),
\]
where $\abs{\lambda_1},\abs{\lambda_2}<1$. The tower property of conditional expectations and the fact that $X_i$ and $Y_i$ are conditionally independent of $\G_i$ given $\F$ imply that
\[
	\E[h_{i1}^{(j)}(0)-h_{i2}^{(j)}(0)\mid \F]=0 \qtext{a.s.}
\]
for $j=1,2$. Finally, using the bounds in \eqref{eq:norm_drv_bounds} and noticing that $\absin{f^{(j)}(x-q)}\le D\epsilon^{-3}x^{3-j}\times \ind_{(q,q+\epsilon)}(x)$ for $1\le j\le 3$, we get
\[
	\absin{\E[h_{i1}^{(3)}-h_{i2}^{(3)}\mid \F]}\le \frac{B}{\epsilon^3}\left(\E[\norm{X_i}^3\mid \F]+\E[\norm{Y_i}^3\mid \F]\right) \qtext{a.s.}
\]
for some constant $B>0$. The result then follows from Lemma 4 in \cite{Rhee/Talagrand:86}, i.e., there is a constant $M>0$ such that $\E[\norm{Y_i}^{3}\mid \F]\le M\E[\norm{X_i}^{3}\mid \F]$ a.s.
\end{subproof}

\medskip

Using Lemma \ref{lemma:chi_bound} when $d\ge 2$ it follows that
\begin{equation}
\label{eq:aux_CLT2}
	\DK{\norm{T},\norm{N}\mid \G}\le \frac{B}{\epsilon^3}\sum_{i=1}^n \E[\norm{X_i}^3\mid \F]+\frac{C_d'}{\sqrt{\upsilon}}\epsilon \qtext{a.s.}
\end{equation}
For $d=1$ we have $\PR{N\in S_{q,\epsilon}\mid \F}\le \epsilon/\sqrt{2\pi \upsilon}$ and the same bound holds. Finally, since \eqref{eq:aux_CLT2} holds for any $\epsilon>0$ a.s., it holds for random $\epsilon$ a.s.\ on $\{\epsilon\in (0,\infty)\}$. Then the result follows by taking $\epsilon=\left(\sqrt{\upsilon}\sum_{i=1}^n\E[\norm{X_i}^3\mid \F]\right)^{1/4}$.
\end{proof}

\medskip

\begin{corollary}
\label{corr:aux_CLT1}
Let $X_1,\ldots,X_n$ be conditionally i.i.d.\ given $\F\subset \H$ with $\E[X_1\mid \F]=0$ and $\E[\norm{X_1}^3\mid \F]<\infty$ a.s. Let $T\eqdef n^{-1/2}\sum_{i=1}^n X_i$ and let $N\mid \F\sim \ND{0}{V}$, where $V=\E[X_1X_1^{\top}\mid \F]$ a.s. Then, assuming that $\upsilon\equiv \lambda_{(d\vee 2-1)}(V)>0$ a.s.
\[
	\DK{\norm{T},\norm{N}\mid \F}\le C_d\left(\frac{\E[\norm{X_1}^3\mid \F]}{\upsilon^{3/2}\sqrt{n}}\right)^{1/4} \qtext{a.s.},
\]
where $C_d>0$ is a constant depending only on $d$.
\end{corollary}

\medskip

\begin{lemma}
\label{lemma:aux_CLT2}
Let $X_1,\ldots,X_n$ be random variables that are conditionally i.i.d.\ given $\F\subset\H$ with $\E[X_1\mid \F]=0$ and $\E[\abs{X_1}^3\mid \F]<\infty$ a.s. Let $T\eqdef n^{-1/2}\sum_{i=1}^n X_i$ and $N\mid \F\sim \ND{0}{\sigma^2}$, where $\sigma^2=\Var(X_1\mid \F)$ a.s. Then, assuming that $\sigma>0$ a.s.,
\begin{equation*}
	\DK{T,N\mid \F}\le C\frac{\E[\abs{X_1}^3\mid \F]}{\sigma^3\sqrt{n}} \qtext{a.s.},
\end{equation*}
where $C>0$ is a constant.
\end{lemma}

\begin{proof}
The proof is similar to the proof of Theorem 11.4.1 in \cite{Athreya:2006:MTP} (\textit{for the unconditional case}) and so is omitted.
\end{proof}

\medskip

\begin{lemma}
\label{lemma:aux_normal_approx1}
Suppose that $\G$ and $\F$ are $\sigma$-fields such that $\F\subset \G\subset\H$, $X$ and $Y$ are random vectors in $\R^d$ such that $X\mid \G\sim \ND{0}{\Sigma_X}$ and $Y\mid \F\sim \ND{0}{\Sigma_Y}$. Then, assuming that $\upsilon\equiv\lambda_{(d\vee 2-1)}(\Sigma_Y)>0$ a.s.,
\begin{align}
\label{eq:aux_normal_approx_bound}
\begin{aligned}
	&\DK{\norm{X},\norm{Y}\mid \G,\F}\le C_d\left(\upsilon^{-1}\norm{\Lambda_X-\Lambda_Y}_{e,\infty}\right)^{1/3} \qtext{a.s.},
\end{aligned}
\end{align}
where $C_d$ is a constant depending only on $d$ and $\Lambda_{(\csdot)}$ is the matrix of eigenvalues corresponding to $\Sigma_{(\csdot)}$.
\end{lemma}

\begin{proof}
Let $f$ be a twice continuously differential function such that $f(x)=1$ if $x\le 0$, $f(x)=0$ if $x\ge \epsilon>0$ and $\abs{f^{(j)}}\le D\epsilon^{-j}\ind_{(0,\epsilon)}(x)$ for some constant $D>0$ and $1\le j\le 2$. Further, set
\[
	g_r(s)\eqdef f(\norm{s}-r).
\]
As in the proof of Theorem \ref{thm:aux_CLT1} for any $\epsilon>0$ w.p.1,
\begin{align*}
	&\DK{\norm{X},\norm{Y}\mid \G,\F} \\
	&\qquad\le \sup_{q\in \Q_{> 0}}\abs{\E[g_q(X)\mid \G]-\E[g_q(Y)\mid \F]}+\sup_{q\in \Q_{\ge 0}}\PR{Y\in S_{q,\epsilon}\mid \F}.
\end{align*}
Let $Z_1$ and $Z_2$ be independent standard normal random vectors in $\R^d$ that are independent of $\G$ and $\F$, respectively. Then
\begin{align*}
	\E[g_q(X)\mid \G]-\E[g_q(Y)\mid \F]&=\E[g_q(\Lambda_X^{1/2}Z_1)\mid \G]-\E[g_q(\Lambda_Y^{1/2}Z_2)\mid \F] \\
	&=h_{q,1}(\Lambda_X)-h_{q,2}(\Lambda_Y) \qtext{a.s.},
\end{align*}
where $h_{q,1}(\lambda)\eqdef\E g_q(\lambda^{1/2}Z_1)$ and $h_{q,2}(\lambda)\eqdef\E g_q(\lambda^{1/2}Z_2)$ \citep[see, e.g.,][Lemma 6.2.1]{Durrett:10:Prob}.

\begin{claim}
There exists a constant $B_d$ depending only on $d$ such that for any $q>0$,
\[
	\abs{h_{q,1}(\lambda_X)-h_{q,2}(\lambda_Y)}\le \frac{B_d}{\epsilon^2}\normin{\lambda_X-\lambda_Y}_{e,\infty}.
\]
\end{claim}
\begin{subproof}
For $t\in [0,1]$ let $Z(t)\eqdef\sqrt{t}\lambda_X^{1/2}Z_1+\sqrt{1-t}\lambda_Y^{1/2}Z_2$ and $\phi(t)\eqdef \E g_q(Z(t))$. Then
\[
	h_{q,1}(\lambda_X)-h_{q,2}(\lambda_Y)=\phi(1)-\phi(0)=\int_0^1 \phi'(t)dt.
\]
Using the integration by parts formula \citep[see Equation A.17 in][Section~A.6]{Talagrand:11:SpinGlasses} for $t\in (0,1)$,
\begin{align*}
	\phi'(t)&=\frac{1}{2}\E\left[\left(\lambda_X^{1/2}Z_1/\sqrt{t}-\lambda_Y^{1/2}Z_2/\sqrt{1-t}\right)^{\top} \nabla g_q(Z(t))\right] \\
	&=\frac{1}{2}\E\left[\vec{i}^{\top}(\lambda_X-\lambda_Y)\circ \nabla^2 g_q(Z(t))\vec{i}\right],
\end{align*}
where $\vec{i}$ is the vector of ones. Therefore,
\[
	\abs{\int_0^1 \phi'(t)dt}\le \normin{\lambda_X-\lambda_Y}_{e,\infty}\int_0^1 \E\abs{\vec{i}^{\top}\nabla^2 g_q(Z(t))\vec{i}}dt.
\]
Since $\absin{f^{(j)}(x-q)}\le D\epsilon^{-2}x^{2-j}\times \ind_{(q,q+\epsilon)}(x)$ for $1\le j\le 2$, the $(k,l)$-th element of the Hessian of $g_q$ is bounded by $D'\epsilon^{-2}$ for some constant $D'>0$. Therefore,
\[
	\abs{h_{q,1}(\lambda_X)-h_{q,2}(\lambda_Y)}\le \frac{D'd^2}{\epsilon^2}\normin{\lambda_X-\lambda_Y}_{e,\infty}. \qedhere
\]
\end{subproof}

\medskip

Using Lemma \ref{lemma:chi_bound} when $d\ge 2$ it follows that
\begin{equation}
\label{eq:aux_normal_approx_bound1}
	\DK{\norm{X},\norm{Y}\mid \G}\le \frac{B_d}{\epsilon^2}\normin{\Lambda_X-\Lambda_Y}_{e,\infty}+\frac{C_d'}{\sqrt{\upsilon}}\epsilon \qtext{a.s.}
\end{equation}
For $d=1$, $\PR{N\in S_{q,\epsilon}\mid \F}\le \epsilon/\sqrt{2\pi \upsilon}$ a.s., so that the same bound holds. Finally, since \eqref{eq:aux_normal_approx_bound1} holds for any $\epsilon>0$ a.s., it holds for random $\epsilon$ a.s.\ on $\{\epsilon\in (0,\infty)\}$. Consequently, the result follows by taking $\epsilon=(\sqrt{\upsilon}\norm{\Lambda_X-\Lambda_Y}_{e,\infty})^{1/3}$ and noticing that \eqref{eq:aux_normal_approx_bound} holds trivially on $\{\norm{\Lambda_X-\Lambda_Y}_{e,\infty}=0\}$.
\end{proof}

\medskip

\begin{lemma}
\label{lemma:aux_normal_approx2}
Suppose that $\G$ and $\F$ are $\sigma$-fields such that $\F\subset \G\subset\H$ and let $X\mid \G\sim \ND{0}{\sigma_X^2}$ and $Y\mid \F\sim \ND{0}{\sigma_Y^2}$. Then, assuming that $\sigma_Y>0$ a.s.,
\[
	\DK{X,Y\mid \G,\F}\le C\abs{\sigma_X^2/\sigma_Y^2-1}^{1/3} \qtext{a.s.},
\]
where $C>0$ is a constant.
\end{lemma}

\begin{proof}
The proof is similar to one for Lemma \ref{lemma:aux_normal_approx1}, and so is omitted.
\end{proof}

\medskip

\begin{lemma}
\label{lemma:aux_normal_approx3}
Let $(G,(Y,X))$ be a network dependent process in $\R\times \R^d$ and let $\F$ be a sub-$\sigma$-field of $\H$ such that:
\begin{itemize}[leftmargin=1.65em]
	\item[\tn{(a)}] $Y$ is conditionally independent of $X$ given $\F$;
	\item[\tn{(b)}] $Y_i$ and $Y_j$ are conditionally independent given $\F$ if $j\notin B_i\eqdef N(i;s)$ for some $s>0$;
	\item[\tn{(c)}] $\D(G)$ is $\F$-measurable.
\end{itemize}
Let $\G\eqdef \sigma(\F\cup \sigma(X))$, $T\eqdef\sum_{i\in N}Y_iX_i$, and $Z\mid \G\sim \ND{0}{V}$, where $V=\E[TT^{\top}\mid \G]$ a.s. Then, assuming that $\upsilon\equiv\lambda_{(d\vee 2-1)}(V)>0$ a.s.,
\[
	\DK{\norm{T},\norm{Z}\mid \G}\le C_d\left(\upsilon^{-3/2}\beta\right)^{1/4} \qtext{a.s.},
\]
where $C_d>0$ is a constant depending only on $d$ and
\[
	\beta\eqdef\sum_{i\in N}\sum_{j\in B_i}\sum_{k\in B_i\cup B_j}\prod_{l\in\{i,j,k\}}\norm{Y_l}_{\F,3}\norm{X_l}_{\infty}.
\]
In addition, when $d=1$,
\[
	\DK{T,Z\mid \G}\le C_1\left(\upsilon^{-3/2}\beta\right)^{1/4} \qtext{a.s.}
\]
\end{lemma}

\begin{proof}
We use the notation from the proof of Theorem \ref{thm:aux_CLT1}. First, for any $\epsilon>0$ w.p.1,
\begin{align*}
	&\DK{\norm{T},\norm{Z}\mid \G} \\
	&\qquad \le \sup_{q\in Q_{>0}}\abs{\E[g_q(T)-g_q(Z)\mid \G]}+\sup_{q\ge 0}\PR{Z\in S_{q,\epsilon}\mid \G}.
\end{align*}
Let $Y'\mid \F\sim\ND{0}{\Sigma}$ conditionally independent of $(Y,X)$ given $\F$, where $\Sigma=\Var(Y\mid \F)$ a.s., and let $Z'\eqdef\sum_{i\in N}Y_i'X_i$. Note that $\E[g_q(Z)\mid \G]=\E[g_q(Z')\mid \G]$ a.s. Also let $\QM_Y$ and $\QM_{Y'}$ be the regular conditional distributions of $Y$ and $Y'$ given $\F$ and $\QM\eqdef\QM_Y\otimes \QM_{Y'}$. Since $X$ is $\G$-measurable, for almost all $\omega\in \Omega$,
\[
	\E[g_q(T)-g_q(Z)\mid \G](\omega)=h_q(\omega),
\]
where
\[
	h_q(\omega)\eqdef\int_{R^{n\times 2}} g_q\left(\sum_{i\in N}y_iX_i(\omega)\right)-g_q\left(\sum_{i\in N}y_i' X_i(\omega)\right)\QM(\omega,d(y,y'))
\]
\citep[see, e.g.,][Theorem~5.4]{Kallenberg:02:FMP}.

\begin{claim}
There exists a constant $B_d>0$ depending only on $d$ such that for any $q>0$,
\[
	\abs{h_q(\omega)}\le \frac{B_d}{\epsilon^3}\sum_{i\in N}\sum_{j\in B_i}\sum_{k\in B_i\cup B_j}\prod_{l\in\{i,j,k\}}(\chi_l(\omega))^{1/3}\norm{X_l(\omega)}_{\infty},
\]
where $\chi_i(\omega)\eqdef \int_{\R^n}y_i^3\QM_Y(\omega,d(y))$.
\end{claim}

\begin{subproof}
For $y\equiv\seq{y_i}_{i\in N}$ let $\phi(y)\eqdef g_q(\sum_{i\in N}y_iX_i(\omega))$. Then the result follows from Theorem 3.4 in \cite{Rollin:13} by observing that
\[
	\norm{\phi_{ijk}}_{\infty}\le \frac{B_d'}{\epsilon^3}\prod_{l\in \{i,j,k\}}\norm{X_l(\omega)}_{\infty}
\]
for some constant $B_d'>0$ depending only on $d$, where $\phi_{ijk}$ is the third order partial derivative of $\phi$ w.r.t. the coordinates $i$, $j$, and $k$.
\end{subproof}

\medskip

As in the proof of Theorem \ref{thm:aux_CLT1} there exists a constant $C_d'>0$ depending only on $d$ such that
\[
	\sup_{q\ge 0}\PR{Z\in S_{q,\epsilon}\mid \F}\le \frac{C_d'}{\sqrt{\upsilon}}\epsilon \qtext{a.s.}
\]
Therefore, noticing that $\chi_i=\E[\abs{Y_i}^3\mid \F]$ a.s., the result follows by taking $\epsilon=\left(\sqrt{\upsilon}\beta\right)^{1/4}$. The second assertion for $d=1$ follows similarly.
\end{proof}

\end{document}